\documentclass{LMCS}

\def\doi{7 (3:17) 2011}
\lmcsheading%
{\doi}
{1--33}
{}
{}
{Jun.~14, 2010}
{Sep.~26, 2011}
{}   
\pdfoutput=1
\newif\iftr\trfalse
\newif\ifdraft\draftfalse

\usepackage{stmaryrd,latexsym,amsmath,amssymb}
\usepackage{graphicx}
\usepackage{listings}
\usepackage{utf8}

\usepackage{enumerate}
\usepackage{hyperref}

\usepackage{wrapfig}
\usepackage{subfig}
\usepackage[T1]{fontenc}
\usepackage{textcomp}
\usepackage{url} 

\usepackage[scaled=0.8]{beramono}

\usepackage{color}
\usepackage{mathpartir}

\ifdraft
\newcommand{\draftnote}[1]{\marginpar{\sf \footnotesize{#1}}}
\else
\newcommand{\draftnote}[1]{}
\fi

\usepackage{declarations}

\lstdefinelanguage{WHILE}{
  keywords={output, input, else,  string, max, mod, div, to, int, while, if, 
    then, newkey, for, on, do, skip, declassify, endorse },
}
\newcommand\Seg{\ensuremath{\mathrm Seg}}

\usepackage{ttquot}
\usepackage{tightheaders}
\iftr
\fi
\begin{document}

\title{Attacker Control and Impact for Confidentiality and Integrity}
\iftr
\author{Aslan Askarov \quad Andrew C. Myers \\
Department of Computer Science,  Cornell University
}
\else
\author[A.~Askarov]{Aslan Askarov}
\address{Department of Computer Science,  Cornell University}
\email{aslan@cs.cornell.edu and andru@cs.cornell.edu}
\author[A.~C.~Myers]{Andrew C. Myers}
\fi

\begin{abstract}
Language-based information flow methods offer a principled way to
enforce strong security properties, but enforcing noninterference
is too inflexible for
realistic applications. Security-typed languages have therefore
introduced declassification mechanisms for relaxing confidentiality
policies, and endorsement mechanisms for relaxing integrity policies.
However, a continuing challenge has been to define what security is
guaranteed when such mechanisms are used.
This paper presents a new semantic framework for expressing security
policies for declassification and endorsement in a language-based
setting. 
The key insight is that security can be characterized in terms of
the influence that declassification and endorsement allow to the attacker.
The new framework introduces two notions of security to
describe the influence of the attacker. Attacker
control defines what the attacker is able to learn from
observable effects of this code; attacker impact captures the
attacker's influence on trusted locations.
This approach yields novel security conditions for checked
endorsements and robust integrity. The framework is flexible enough to
recover and to improve on the previously introduced notions of
robustness and qualified robustness.  Further, the new security
conditions can be soundly enforced by a security type system. The
applicability and enforcement of the new policies is illustrated
through various examples, including data sanitization and
authentication.

\end{abstract}


\maketitle

\section{Introduction}

Many common security vulnerabilities can be seen as violations of
either confidentiality or integrity. As a general way to prevent these
information security vulnerabilities, information flow control has
become a popular subject of study, both at the language
level~\cite{Sabelfeld:Myers:JSAC} and at the operating-system
level (e.g.,~\cite{MR92,asbestos,dstar}).
The language-based approach holds the appeal that
the security property of noninterference~\cite{Goguen:Meseguer:Noninterference},
can be provably enforced using a type
system~\cite{Volpano:Smith:Irvine:Sound}.
In practice, however, noninterference is too rigid: many programs considered
secure need to violate noninterference in limited ways.

Using language-based downgrading mechanisms such as 
_declassification_~\cite{ml-ifc-97,pottier00}
and _endorsement_~\cite{Oerbaek:Palsberg:Trust,zznm02},
programs can be written in which
information is intentionally released, and in which untrusted
information is intentionally used to affect trusted information
or decisions. Declassification relaxes confidentiality policies,
and endorsement relaxes integrity policies.  Both endorsement and
declassification have been essential for building realistic
applications, such as various applications
built with Jif~\cite{Myers:POPL99,jif}: games~\cite{as05},
a voting system~\cite{Clarkson:Chong:Myers:Oakland08},
and web applications~\cite{Chong+:SOSP07}.

A continuing challenge is to understand what security is obtained when
code uses downgrading. This paper contributes
a more precise and satisfactory answer to this question, particularly
clarifying how the use of endorsement weakens confidentiality.
While much
work has been done on declassification (usefully summarized by Sands and
Sabelfeld~\cite{Sabelfeld:Sands:JCS}),
there is comparatively little work on the
interaction between confidentiality and endorsement.

To see such an interaction,
consider the following notional code example, in which a service
holds both old data ("old_data") and new data ("new_data"), but
the new data is not to be released until time "embargo_time".
The variable "new_data" is considered confidential,
and must be declassified to be released:

\begin{lstlisting}
if request_time >= embargo_time
  then return declassify(new_data)
  else return old_data
\end{lstlisting}

\noindent
Because the requester is
not trusted, the requester must be treated as a possible attacker.
Suppose the requester has control over the variable "request_time",
which we can model by considering that variable to be low-integrity.
Because the intended security policy
depends on "request_time", the attacker controls the
policy that is being enforced, and can obtain the confidential new
data earlier than intended. This example shows that the integrity of
"request_time" affects the confidentiality of "new_data". Therefore,
the program should be considered secure only when the guard expression,
"request_time >= embargo_time", is high-integrity.

A different but reasonable security policy is that the requester may
specify the request time as long as the request time is in the past.
This policy could be enforced in a language with endorsement by first
checking the low-integrity request time to ensure it is in the past;
then, if the check succeeds, endorsing it to be high-integrity and
proceeding with the information release. The explicit endorsement is
justifiable because the attacker's actions are permitted to affect the
release of confidential information as long as adversarial inputs have
been properly sanitized. This is a common pattern in servers that
process possibly adversarial inputs.

_Robust declassification_ has been introduced in prior
work~\cite{zm01b,Myers:Sabelfeld:Zdancewic:JCS06,Chong:Myers:CSFW06} as a
semantic condition for
secure interactions between integrity and confidentiality.  The prior
work also develops type systems for enforcing robust declassification,
which are implemented as part of Jif~\cite{jif}. However, prior
security conditions for robustness are not satisfactory, for two reasons.
First, these
prior conditions characterize information security only for
terminating programs. A program that does not terminate is automatically
considered to satisfy robust declassification, even if it releases
information improperly during execution. Therefore
the security of programs that do not terminate, such as servers,
cannot be described.
A second and perhaps even more serious limitation is that prior
security conditions largely ignore the
possibility of endorsement, with the exception of _qualified
robustness_~\cite{Myers:Sabelfeld:Zdancewic:JCS06}. Qualified 
robustness gives the "endorse"
operation a somewhat ad-hoc, nondeterministic semantics, to
reflect the attacker's ability to choose the
endorsed value. This approach 
operationally models what the attacker can do, but does
not directly describe the attacker's control over
confidentiality. The introduction of nondeterminism also makes the
security property possibilistic. However, possibilistic security
properties have been criticized because they can weaken under
refinement~\cite{Roscoe95,SV98}.

The main contribution of this paper is a general, language-based
semantic framework for expressing information flow security and
semantically capturing the ability of the attacker to influence both
the confidentiality and integrity of information. The key building
blocks for this semantics are _attacker knowledge_~\cite{Askarov:Sabelfeld:SP07}
and its (novel) dual, _attacker impact_, which respectively
describe what attackers can know and what
they can affect. Building upon attacker knowledge, the interaction of
confidentiality and integrity, which we term _attacker control_, can
be characterized formally. The robust interaction of confidentiality
and integrity can then be captured cleanly as a constraint on attacker
control. Further, endorsement is naturally represented in this
framework as a form of attacker control, and a more satisfactory
version of qualified robustness can be defined.
All these security conditions can be
formalized in both _progress-sensitive_ and _progress-insensitive_
variants, allowing us to describe the security of both terminating
and nonterminating systems.

We show that the progress-insensitive variants of these improved
security conditions are enforced soundly by a simple security type
system.  Recent versions of Jif have added a _checked
endorsement_ construct that is useful for expressing complex security
policies~\cite{Chong+:SOSP07}, but whose semantics were not precisely
defined; this paper gives semantics, typing rules and a semantic
security condition for checked endorsement, and shows that checked endorsement
can be translated faithfully into simple endorsement at both the
language and the semantic level.
Our type system can easily be adjusted to enforce the progress-sensitive variants of the security conditions, as has been shown in the literature~\cite{Volpano:Smith:Probabilistic:CSFW,ONeil+:CSFW06}.

The rest of this paper is structured as follows.
Section~\ref{sec:semantics} shows how to define information security
in terms of attacker knowledge. Section~\ref{sec:attacks} introduces
attacker control. Section~\ref{sec:robust} defines 
progress-sensitive and progress-insensitive robustness using the new
framework. Section~\ref{sec:endorsement} extends this to improved
definitions of robustness that allow endorsements, generalizing
qualified robustness. A type system for enforcing these robustness
conditions is presented in Section~\ref{sec:enforcement}.
The checked endorsement construct appears in
Section~\ref{sec:checked}, which introduces a new notion of
robustness that allows checked endorsements, and shows that it
can be understood in terms of robustness extended
with simple endorsements. Section~\ref{sec:attacker-impact} introduces
attacker impact.
Additional examples are presented in Section~\ref{sec:examples},
related work is discussed in Section~\ref{sec:related}, and
Section~\ref{sec:conclusion} concludes.

This paper is an extended version of a previous paper by the same
authors~\cite{am10}. The significant changes
include proofs of all the main theorems, a semantic rather than
syntactic definition of fair attacks, and a renaming of
``attacker power'' to ``attacker impact''.
\section{Semantics }
\label{sec:semantics}
 
\paragraph{Information flow levels}
We assume two security levels for confidentiality --- _public_ and
_secret_ --- and two security levels for integrity --- _trusted_ and
_untrusted_. These levels are denoted respectively $\Public, \Secret$
and $\Trusted, \Untrusted$.  We define information flow ordering
$\sqsubseteq$ between these two levels: $\Public \sqsubseteq \Secret$,
and $\Trusted \sqsubseteq \Untrusted$. The four levels define a
security lattice, as shown on Figure~\ref{fig:lattice}. Every point on
this lattice has two security components: one for confidentiality, and
one for integrity.
We extend the information flow ordering to elements on this lattice: $\level_1 \sqsubseteq \level_2$ if the ordering holds between the corresponding components. 
As is standard, we define _join_ $\level_1 \join \level_2$  as the
least upper bound of $\level_1$ and $\level_2$, and _meet_ $\level_1
\sqcap \level_2$ as the greatest lower bound of $\level_1$ and
$\level_2$. 
All four lattice elements are meaningful; for example, it is possible
for information to be both secret and untrusted when it depends on
both secret and untrusted (i.e., attacker-controlled) values.
This lattice is the simplest possible choice for exploring the topics
of this paper; however, the results of this paper
straightforwardly generalize to the richer security lattices used in other work on robustness~\cite{Chong:Myers:CSFW06}.

\begin{figure}[t] %
\begin{minipage}{0.41\textwidth}
\begin{center}
\includegraphics{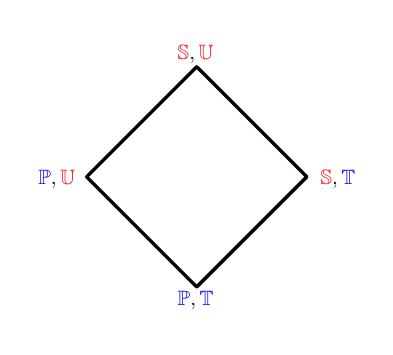}
\end{center}
\caption{Information flow lattice}
\label{fig:lattice}
\end{minipage}
\begin{minipage}{0.55\textwidth}
\vspace{9.5ex}
\begin{align*}
e ::=\; & n\ |\ x\ |\ e\ \mathit{op}\ e\ %
 \\
c::=\; & \Skip\ |\ x := e \ |\ c;c\ \\ &|\ \ifThenelse{e}{c_1}{c_2}\   \ |\  \while{e}{c}\    %
\end{align*}
\caption{Syntax of the language}
\label{fig:syntax}
\end{minipage}
\end{figure}

\paragraph{Language and semantics}

\begin{figure}
\begin{mathpar}
\inferrule{}{\configTwo{n}{m} \downarrow {n}} 
\and
\inferrule{}{\configTwo{x}{m} \downarrow {m(x)}}
\and
\inferrule{\configTwo{e_1}{m} \downarrow {v_1} \and 
  \configTwo{e_2}{m} \downarrow {v_2} \and \
  v = v_1\ \mathbf{op}\ v_2
} {
  \configTwo{e_1\ \mathit{op}\ e_2}{m} \downarrow {v} }
\end{mathpar}
\caption{Semantics of expressions}
\label{fig:sem:expressions}
\begin{mathpar}
   \inferrule{}{ \configTwo{\Skip}{m} \labarrow{}{}  \configTwo{\halt}{m}}
\and 
    \inferrule
{\configTwo{e}{m} \downarrow 
  {v}
}{\configTwo{x:=e}{m}\labarrow{}{(x,v)} \configTwo{\halt}{m[x\mapsto v]}}
\and
\inferrule{ \configTwo{c_1}{m} \labarrow{}{t} \configTwo{c'_1}{m'} 
}
{\configTwo{c_1; c_2}{m} \labarrow{}{t} \configTwo{c'_1; c_2}{m'}}
\and
\inferrule{ \configTwo{c_1}{m} \labarrow{}{t} \configTwo{\halt}{m'} 
}{\configTwo{c_1; c_2}{m}  \labarrow{}{t} \configTwo{c_2}{m'}}
\and
\inferrule{ \configTwo{e}{m} \downarrow {n}  \\ n \neq 0}{
  \configTwo{\ifThenelse{e}{c_1}{c_2}}{m} \labarrow{}{} \configTwo{c_1}{m}}
\and
\inferrule{ \configTwo{e}{m} \downarrow {n}
 \\ n = 0}{\configTwo{\ifThenelse{e}{c_1}{c_2}}{m} \labarrow{}{} \configTwo{c_2}{m}}
\and
\inferrule{  \configTwo{e}{m} \downarrow {n}  \\ n \neq 0}{
  \configTwo{\while{e}{c}}{m} \labarrow{}{} \configTwo{c; \while{e}{c}}{m}}
\and
\inferrule{  \configTwo{e}{m} \downarrow {n}  \\ n = 0}{
  \configTwo{\while{e}{c}}{m} \labarrow{}{} \configTwo{\halt}{m}}
\end{mathpar}
\caption{Semantics of commands}
\label{fig:sem:commands}
\end{figure}

We consider a simple imperative language with syntax presented in Figure~\ref{fig:syntax}.
The semantics of the language is fairly standard and is given 
in  Figures~\ref{fig:sem:expressions} and~\ref{fig:sem:commands}.
For expressions, we define big-step evaluation of the form 
 $\configTwo{e}{m} \downarrow v$, where $v$ is the result of evaluating expression $e$ in memory $m$. 
For commands, we define a small-step operational semantics, in which 
 a single transition is written as $\configTwo{c}{m} \labarrow{}{t} \configTwo{c'}{m'}$,
where $c$ and $m$ are the initial command and memory, and $c'$ and $m'$ are the resulting command and memory.
The only unusual feature is the annotation $t$ on
each transition, which we call an _event_.
Events record assignments: an assignment to variable $x$ of value $v$ is recorded by an event $(x,v)$. 
This corresponds to our attacker model, in which the attacker may only observe assignments to public variables.
We write $\configTwo{c}{m}
\labarrowstar{}{\vec t}  $ to mean that trace  $\vec t$ is produced starting from
$\configTwo{c}{m}$ using zero or more transitions.
Each trace $\vec t$ is composed of individual events $t_1 \tconcat t_2 \cdots t_k \cdots$, and 
a \emph{prefix} of $\vec t$  up to the $i$-th event  is denoted
as~$\vec{t_i}$; we use the operator $\;\tconcat\;$ to denote the concatenation of two traces or events.
If a transition does not affect memory, its event is _empty_, which is
either written as ε or is omitted, e.g.: 
$\configTwo{c}{m} \labarrow{}{} \configTwo{c'}{m'}$. 

Finally, we assume that the _security environment_ $\G$ maps program
variables to their security levels.
Given a memory $m$, we write $m_\Public$ for the public part of
the memory; similarly, $m_\Trusted$ is the trusted part of $m$. 
We write $m =_\Trusted m'$ when memories $m$ and $m'$ agree on their trusted parts, and $m =_\Public m'$ when $m$ and $m'$ agree on their public parts.

\subsection{Attacker knowledge}

This section provides background on the attacker-centric model for
information flow security~\cite{Askarov:Sabelfeld:SP07}. We recall
definitions of attacker knowledge, progress knowledge, and
divergence knowledge, and introduce
progress-(in)sensitive \emph{release events}.

\paragraph{Low events} Among the events that are generated during a trace, we distinguish a sequence of
low (or public) events.  Low events correspond to observations that an
attacker can make during a run of the program. We assume that the attacker
may observe individual assignments to public variables. Furthermore,
if the program terminates, we assume that a termination event
$\Downarrow$ may also be observed by the attacker. 
If attacker can detect divergence of programs (cf. Definition~\ref{def:knowledge:divergence}) then divergence $\Uparrow$ is also a low event.

Given a trace $\vec t$, low events in that trace are denoted as $\vec t_\Public$. A single low event is often denoted as $\ell$, and a sequence of low events is denoted as $\vec \ell$. We overload the notation for semantic transitions,  writing $\configTwo{c}{m} \labarrowstar{}{\vec \ell}$ if only low events produced from configuration $\configTwo{c}{m}$ 
are relevant; that is, there is a trace $\vec t$ such that $\configTwo{c}{m} \labarrowstar{}{\vec t} \land \vec t_\Public = \vec \ell$. 
Low events are the key element in the definition of \emph{attacker knowledge}~\cite{Askarov:Sabelfeld:SP07}.

The knowledge of the attacker is described by the set of initial memories compatible with low observations. Any reduction in  this set means the attacker has learned something about secret parts of the initial memory.

\begin{defi}[Attacker knowledge]
Given a sequence of low events $\vec \ell$, initial low memory $m_{\Low}$, and program~$c$,
 \emph{attacker knowledge} is
\[
k(c, m_{\Low}, \vec \ell) \defn \{ m' \ |\ m_\Low = m'_\Low \land  \configTwo{c}{m'} \labarrowstar{}{\vec \ell}  \}
\]
\end{defi}\vspace{6 pt}

\noindent Attacker knowledge gives a handle on what information the attacker learns with every low event.  The
smaller the knowledge set, the more precise is the attacker's information about secrets. 
Knowledge is monotonic in the number of low events: as the program
produces low events, the attacker may learn more about secrets.

Two extensions of attacker knowledge are useful: \emph{progress knowledge}~\cite{Askarov+:Termination,Askarov:Sabelfeld:CSF09}
and \emph{divergence knowledge}~\cite{Askarov+:Termination}.
\begin{defi}[Progress knowledge]
\label{def:knowledge:progress}
Given a sequence of low events $\vec \ell$, initial low memory $m_\Low$, and a program  $c$, 
define \emph{progress knowledge} $k_\rightarrow(c, m_\Low, \vec \ell) $ as 
\[
k_\to (c, m_\Low, \vec \ell) \defn \{ m' \ |\ m'_\Public = m_\Public \land \exists \ell'~.~
 \configTwo{c}{m'} \labarrowstar{}{\vec \ell} \configTwo{c''}{m''}
 \labarrowstar{}{\ell'}  \}
\]
\end{defi}\vspace{6 pt}

\noindent Progress knowledge represents the information the attacker 
obtains by seeing public events $\vec\ell$ followed by some other
public event.  Progress knowledge and attacker knowledge are related
as follows:
given a program $c$, memory $m$ and a sequence of  low events $\ell_1
\cdots \ell_{n}$ obtained from $\configTwo{c}{m}$, we have that for all $i < n$,
\[
k(c, m_\Public, \vec \ell_i) \supseteq k_\to (c, m_\Public, \vec \ell_i) \supseteq k(c, m_\Public, \vec \ell_{i+1})
\]
To illustrate this with an example,
consider program 
$\mathit{l} := 0; (\while{h = 0}{\Skip}); \mathit{l} := h$
with initial memory $m(h)=7$. This program produces a sequence of two
low events $(l,0) \tconcat (l,7)$.
The knowledge after the first event $k (c, m_\Public, (\mathit{l}, 0))$ is a set of all possible memories that agree with $m$ on the public parts and can produce the low event $(\mathit l, 0)$. Note that no low events are possible after the first assignment unless $h$ is non-zero. Progress knowledge reflects this: $k_\to (c, m_\Public, (\mathit{l}, 0))$ is a set of memories such that $h\neq 0$. Finally, the knowledge after two events 
$k (c, m_\Public, (\mathit{l}, 0) \tconcat (\mathit{l}, 7))$ is a set of memories where $h=7$. 

Using attacker knowledge, one can express many confidentiality policies~\cite{Banerjee+:SP08,Askarov:Sabelfeld:CSF09,Broberg:Sands:PLAS09}. For example,
a strong notion of \emph{progress-sensitive noninterference}~\cite{Goguen:Meseguer:Noninterference} can be expressed by demanding that knowledge between low events does not change: 
\[k(c, m_\Public, \vec \ell_i) = k(c, m_\Public, \vec \ell_{i+1})\]
Progress knowledge enables expressing more permissive policies, such
as \emph{progress-insensitive
  noninterference}, which allows leakage of information, but only via
termination channels~(in \cite{Askarov+:Termination} it is called
\emph{termination-insensitive}). This is expressed by requiring
equivalence of the progress knowledge after seeing $i$ events with the knowledge obtained after $i+1$-th event: 
 \[k_\to(c, m_\Public, \vec \ell_i) = k(c, m_\Public, \vec \ell_{i+1})\]
In the example $\mathit{l} :=0 ; (\while{h = 0}{\Skip}); \mathit{l} := 1$, the knowledge inclusion between the two events is strict: 
$k(c, m_\Public, (l,0)) \supset k(c, m_\Public, (l,0)\tconcat (l,1))$.
 Therefore, the example does not satisfy progress-sensitive noninterference. On the other hand,
the low event that follows the $\cod{while}$ loop does not reveal more information than the knowledge about the existence of that event. Formally, 
$k_\to(c, m_\Public, (l,0)) = k(c, m_\Public, (l,0)\tconcat(l,1))$, hence the program satisfies progress-insensitive noninterference. 

These definitions also allow us to reason about knowledge changes along \emph{parts of the traces}. We say 
that knowledge is preserved in a progress-(in)sensitive way along a part of a trace, assuming that the respective
knowledge equality holds for the low events that correspond to that
part.

Next, we extend possible observations to a divergence event $\Uparrow$
(we write $\configTwo cm\Uparrow$ to mean configuration $\configTwo cm$ diverges). For attackers that can observe program divergence $\Uparrow$, we
define knowledge on the sequence of low events that includes divergence:
\begin{defi}[Divergence knowledge]
\label{def:knowledge:divergence}
\[
k (c, m_\Low, \vec \ell \Uparrow) \defn \{ m' \ |\ m'_\Public = m_\Public \land
 \configTwo{c}{m'} \labarrowstar{}{\vec \ell} \configTwo{c''}{m''}
\land \configTwo{c''}{m''} \Uparrow \}
\]
\end{defi}\vspace{6 pt}

\noindent Note that the above definition does not require divergence immediately after $\vec \ell$ --- it allows for more low events to be produced after $\vec \ell$. Divergence
knowledge is used in Section~\ref{sec:releasecontrol}.

Let us consider events at which  knowledge preservation is broken. We call these events \emph{release events}.
\begin{defi}[Release events]
  Given a program $c$ and a memory $m$, such that 
\[\configTwo{c}{m} \labarrowstar{}{\vec \ell}
  \configTwo{c'}{m'} \labarrowstar{}{r}\] 
\begin{iteMize}{$\bullet$}
\item  $r$ is a \emph{progress-sensitive release event}, if  $k(c,m_\Public, \vec \ell) \supset k(c,m_\Public, \vec \ell \tconcat  r)$
\item  $r$ is a \emph{progress-insensitive release event}, if  $k_\to(c,m_\Public, \vec \ell) \supset k(c,m_\Public, \vec \ell \tconcat r)$
\end{iteMize}
\end{defi}\vspace{6 pt}

\noindent It is easy to validate that a progress-insensitive release
event is also a progress-sensitive event.  For example, in the program
$ \mathit{low}:=1; \mathit{low'}:=h$, the second assignment is both a
progress-sensitive and a progress-insensitive release event. The
reverse is not true --- in the program $\while{h=0}{\Skip};
\mathit{low}:=1$ the assignment to $\mathit{low}$ is a
progress-sensitive release event, but is not a progress-insensitive
release event.

\section{Attacks}
\label{sec:attacks}
To reason about program security  in the presence of active attacks, we introduce a formal model
of the attacker.  Our formalization follows that in~\cite{Myers:Sabelfeld:Zdancewic:JCS06}, where
attacker-provided code can be injected into the program.
This section provides  examples of how attacker-injected code may affect attacker
knowledge, followed by a semantic characterization of the attacker's influence on knowledge. 

First, we extend the syntax to allow execution of attacker-controlled code:
\begin{align*}
c[\vbullet] :: =\; & \dots \ |\ [\bullet]
\end{align*}

Next, we introduce notation $[\vec t]$ to highlight that the trace $\vec t$ is produced by
attacker-injected code.  
The semantics of the language is extended accordingly.
\begin{mathpar}
  \inferrule{ \configTwo{a}{m} \labarrow{}{t} \configTwo{a'}{m'}}{ \configTwo{[a]}{m} \labarrow{}{[t]} \configTwo{[a']}{m'}}
\and
  \inferrule{}{ \configTwo{[\halt]}{m} \labarrow{}{} \configTwo{\halt}{m}}
\end{mathpar}
We limit 
attacks that can be substituted into holes to so-called \emph{fair attacks}, which represent
reasonable limitations on the impact of the attacker.
Unlike earlier approaches,
where fair attacks are defined syntactically~\cite{Myers:Sabelfeld:Zdancewic:JCS06, Chong:Myers:CSFW06}, 
we define them semantically. This allows us to include a larger set of attacks. To ensure that we include all syntactic attacks we make use of
a reachability translation, explained below.

Roughly, we require a fair attack to not give new knowledge and to not modify trusted variables. 
A refinement of this idea is that an attack is fair if it gives new knowledge but only because
the reachability of the attack depends on a secret.
To capture this refinement,
we define an auxiliary translation to make reachability of attacks explicit.
We assume a trusted, public variable "reach" that does not appear in the source of
$c[\vbullet]$. Let operator $\treach$ be a source-to-source transformation of $c[\vbullet]$ 
that makes reachability of attacks explicit.
\begin{defi}[Explicit reachability translation]
  Given a program $c[\vbullet]$, define $(\treach(c[\vbullet]) $ as follows: 

\begin{iteMize}{$\bullet$}
\item 
$\treach([\bullet]) \Longrightarrow "reach":="reach" + 1; [\bullet]$
\item 
$\treach(c_1; c_2) \Longrightarrow \treach(c_1); \treach(c_2)$ 
\item 
$\treach(\ifThenelse{e}{c_1}{c_2})  \Longrightarrow \ifThenelse{e}{\treach(c_1)}{\treach(c_2)}$
\item 
$\treach(\while{e}{c}) \Longrightarrow \while{e}{\treach(c)}$
\item 
$\treach(c) \Longrightarrow c$ for all other commands $c$
\end{iteMize}

\end{defi}

The formal definition uses that any trace $\vec t$ can be represented as a sequence of subtraces
$\vec t_1 \tconcat [\vec t_2] \cdots 
\vec t_{2*n-1} \tconcat [\vec t_{2*n}]$, where even-numbered subtraces correspond to the events produced by attacker-controlled code. 

Given a trace $\vec t$, we denote the trusted events in the trace as $\vec t_\Trusted$. We use notation $t_\star$  for a single trusted event, and $\vec t_\star$
for a sequence of trusted events.

\begin{defi}[Fair attack]
\label{def:fairattack}
Given a program $c[\vbullet]$, such that $\treach(c[\vbullet]) \Longrightarrow c_\treached[\vbullet]$, 
say that $\vec a$ is a \emph{fair attack} on $c[\vbullet] $ if for all memories $m$,
such that $\configTwo{c_\treached[\vec a]}{m} \labarrowstar{}{\vec t}$ and 
 $\vec t = \vec t_1 \tconcat [\vec t_2] \cdots \vec t_{2*n-1} \tconcat [\vec t_{2*n}]$, i.e., there are $2n$ intermediate configurations
$\configTwo{c_j}{m_j}$,  $1 \leq j \leq 2n$, for which
\[
\configTwo{c_\treached[\vec a]}{m} \labarrowstar{}{\vec t_1} \configTwo{c_1}{m_1}
 \labarrowstar{}{[\vec t_2]}
\configTwo{c_2}{m_2}  \labarrowstar{}{\vec t_3} 
\dots 
 \labarrowstar{}{[\vec t_{2n}]}
\configTwo{c_{2n}}{m_{2n}} \dots 
\]
then for all $i$, $ 1 \leq i \leq n$, it holds that
$
k (c_\treached[\vec a], m, \vec t_1 \cdots \vec t_{2i-1}) 
=
k (c_\treached[\vec a], m, \vec t_1 \cdots  [\vec t_{2i}]) 
$ and $\vec{t_{2i}}_\star=\epsilon$. 

\end{defi}

For example, in the program $\ifThenelse{h>0}{[\bullet]}{\Skip}$ the attacks $a_1 = [\mathit{low}:=1]$
and attack $a_2 = [\mathit{low}:=h > 0]$ are fair, but attack $a_3 = [\mathit{low}:=h]$ is not.

\subsection{Examples of attacker influence}
This section presents a few examples of attacker influence on knowledge. We also introduce pure
availability attacks and progress attacks, to which we refer later in this section. 

In the examples below, 
we use notation $[(u,v)]$ when
a low event $(u,v)$ is generated by attacker-injected code. 

Consider program $[\bullet]; \mathit{low} := u > h; $ where $h$ is a secret variable, and $u$ is
an untrusted public variable. The attacker's code executes before the low assignment and may
change the value of $u$. Consider memory $m$, where $m(h)=7$, and the two attacks $a_1 = u := 0$ and $a_2 =
u:=10$. These attacks result in different values being assigned to variable $\mathit{low}$. The first
trace results in low events $ [(u,0)] \tconcat (\mathit{low}, 0)$, while the second trace results in low
events $[(u,10)] \tconcat (\mathit{low}, 1)$. Therefore, the knowledge about the secret is
different in each trace. We have
\begin{align*}
& k(c[a_1], m_\Public, [(u,0)] \tconcat (\mathit{low}, 0)) = \{ m' \ |\ m'(h) \geq 0 \} \\
& k(c[a_2], m_\Public, [(u,10)] \tconcat (\mathit{low}, 1)) = \{ m' \ |\ m'(h) < 10 \}   
\end{align*}
Clearly, this program gives the attacker some control over what information about secrets he
learns.  Observe that it is not necessary for the last assignment to differ in order for the
knowledge to be different. For example, consider attack $a_3 = u:=5$. This attack results in low
events $[(u, 5)] \tconcat (\mathit{low}, 0)$, which do the same assignment to $\mathit{low}$ as
$a_1$ does.  Attacker knowledge, however, is different from that obtained by $a_1$:
\[
 k(c[a_3], m_\Public, [(u,5)] \tconcat (\mathit{low}, 0)) = \{ m' \ |\ m'(h) \geq 5 \} \\
\]\vspace{2 pt}

\noindent Next, consider program $[\bullet]; \mathit{low} := h$. This program gives away knowledge about the value
of $h$ independently of untrusted variables. The only way for the attacker to influence what
information he learns is to prevent that assignment from happening at all, which, as a result, will
prevent him from learning that information. This can be done by an attack such as $a =
\while{\cod{true}}{\Skip}$, which makes the program diverge before the assignment is reached.
We call attacks like this \emph{pure availability attacks}.
Another example of a pure availability attack is in the program $[\bullet]; (\while{u = 0}{\Skip}) ;
\mathit{low}:= h$.  In this program, any attack that sets $u$ to 0 prevents the assignment from
happening.

Consider another example: $[\bullet]; \while{u < h'}{\Skip}; \mathit{low}:= 1$. As in the previous
example, the value of $u$ may change the reachability of $\mathit{low}:=1$. Assuming the attacker can observe divergence, this is not a
pure availability attack, because diverging before
the last assignment gives  the attacker additional secret information, namely that $ u <
h'$. New information is also obtained if the attacker sees the low assignment.  We name
attacks like this \emph{progress attacks}. In general, a progress attack is an attack that leads
to program divergence in a way that observing that divergence (i.e., detecting
there is no progress) gives new knowledge to the attacker.

\subsection{Attacker control}

We represent attacker control as a set of attacks that are similar in their
influence on knowledge. Intuitively, if a program leaks no information to
the attacker, the control  corresponds to all possible attacks. In general,  the more attacks are similar, the less influence the attacker has. Moreover, the  
control is a temporal property and depends on the trace that has been currently produced.
The longer a trace is, the more influence an attack may have, and the smaller the control set is.

\paragraph{Similar attacks} The key element in the definition of control is specifying when
two attacks are similar. Given a program $c[\vbullet]$, memory $m$, consider two attacks $\vec
a$ and $\vec b$ that produce traces $\vec t$ and $\vec q$ respectively:
\[
\configTwo{c[\vec a]}{m} \labarrowstar{}{\vec t}
\mbox{\ \  \ and\ \ \  } 
\configTwo{c[\vec b]}{m} \labarrowstar{}{\vec q}
\]
We compare~$\vec a$ and~$\vec b$ based on how they change attacker knowledge along their
respective traces. First, if knowledge is preserved
along a subtrace of one of the traces, say $\vec t$, it must be preserved along a subtrace of $\vec q$ as well. Second, 
if at some point in $\vec
t$ there is a release event $(x,v)$, there must be a matching low event $(x,v)$ in $\vec q$, and
the attacks are similar along the rest of the traces. 

Visually, this requirement is described by the two diagrams in Figure~\ref{fig:similarattacks}. 
Each diagram shows the change of knowledge
as more low events are produced. Here the
$x$-axis corresponds to low events, and the $y$-axis reflects the
attacker's uncertainty about initial secrets. Whenever one of the traces 
reaches a release event, depicted by vertical drops, there must be a corresponding low event in the other trace, such that the 
two events agree. This is depicted by the dashed lines between the two diagrams. 

\begin{figure}
\centering
\includegraphics[trim = 0 10 0 10, clip=true]{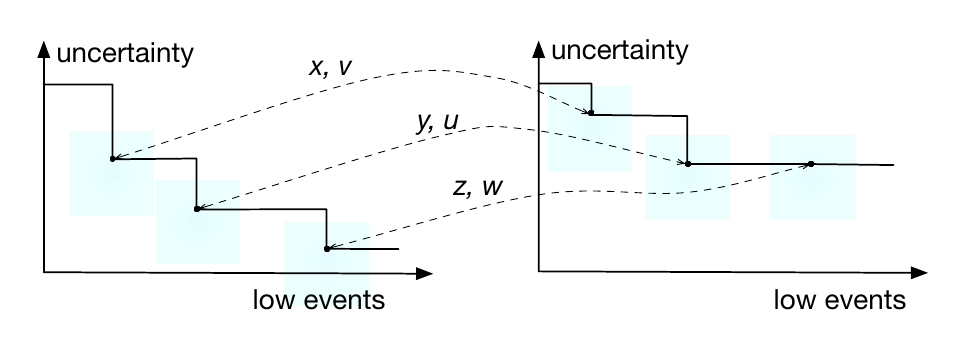}
\caption{Similar attacks and traces}
\label{fig:similarattacks}
\end{figure}

Formally, these requirements are stated using the following definitions. 

\begin{defi}[Knowledge segmentation] 
\label{def:segmentation}
Given a program $c$, memory $m$, and a trace $\vec t$,
 a sequence of indices $p_1 \dots p_N$ such that $p_1 < p_2 \dots < p_N$ and 
$\vec t_\Public = \ell_{1 \dots p_1} \tconcat \ell_{p_1+1  \dots p_2}  \cdots   \ell_{p_{N-1}+1 \dots p_N}$
is called 
\begin{iteMize}{$\bullet$}
\item \emph{progress-sensitive knowledge segmentation} of size $N$, if \\
$\forall j \leq N, \forall i~.~ p_{j-1}+1 \leq i < p_j~.~k(c, m_{\Public}, \vec \ell_i) = k(c, m_\Public, \vec \ell_{i+1})$,
denoted by \\ $\Seg(c,m,\vec t, p_1 \dots p_N)$. 

\item \emph{progress-insensitive knowledge segmentation} of size $N$ if \\
$\forall j \leq N, \forall i~.~ p_{j-1}+1 \leq i < p_j~.~ k_\to(c, m_\Public, \vec \ell_i) = k(c, m_\Public, \vec \ell_{i+1})$,
denoted by \\ $\Seg_\to(c,m,\vec t, p_1 \dots p_N)$. 

\end{iteMize}
Low events $p_i + 1$ for $ 1  \leq i < N$ are called \emph{segmentation events}.
\end{defi}
Note that given a trace, there can be more than one way to segment it, and for every trace consisting of $n$ low events, this can be trivially
achieved by a segmentation of size $n$. 
We use knowledge segmentation to define \emph{attack similarity}:
\begin{defi}[Similar attacks and traces $\sim^{c[\vbullet], m}$
] \label{def:similarattacks}
Given a program $c[\vbullet]$, memory $m$, and two attacks $\vec a$ and
$\vec b$ that produce traces $\vec t$ and $\vec q$, define $\vec a$ and $\vec
b$ as \emph{similar} along $\vec t$ and $\vec q$ for the
progress-sensitive attacker, if there are two segmentations  $ p_1 \dots
p_N$ and $ p'_1 \dots p'_N$ (for some $N$) such
that 
\begin{iteMize}{$\bullet$}
\item
$\Seg(c[\vec a], m, \vec t, p_1 \dots p_N)$, 
\item 
$\Seg(c[\vec b], m, \vec q, p'_1 \dots p'_N)$, and
\item
$\forall i~.~ 1 \leq i < N~.~ t_\Public_{p_i+1} = q_\Public_{p'_i+1}$.
\end{iteMize}
For the progress-insensitive attacker,
the definition is similar except that it uses
progress-insensitive segmentation $\Seg_\to$. 
If two attack--trace pairs are similar, we write
$(\vec a, \vec t) \sim^{c[\vbullet], m} (\vec b, \vec q)$ 
(for progress-insensitive similarity,
$(\vec a, \vec t) \sim^{c[\vbullet], m}_\to (\vec b, \vec q))$.
\end{defi}
The construction of Definitions~\ref{def:segmentation} and~\ref{def:similarattacks} can be illustrated
by program
\[
[\bullet]; \ifThenelse{u}{ ( \while{h \leq 100}{\Skip} )}{\Skip}; \mathit{low}_1 := 0; \mathit{low}_2:=
h > 100
\] 
Consider memory with $m(h) = 555$, and two attacks $a_1 = u:=1$, and $a_2 = u:=0$. 
Both attacks reach the assignments to low variables. However, for $a_2$ the assignment to
$\mathit{low}_2$ is a progress-insensitive release event, while for $a_1$ the knowledge changes
at an earlier assignment.

\paragraph{Attacker control}
We define attacker control with respect to an attack $\vec a$ and a trace $\vec t$ as the set of attacks that are similar to the given attack in its influence on knowledge.

\begin{defi}[Attacker control (progress-sensitive)]
\[
R(c[\vbullet], m, \vec a, \vec t) \defn  
\{ \vec b
\ |\ \exists\vec q~.~(\vec a, \vec t) \sim^{c[\vbullet], m} (\vec b, \vec q)
\}
\]
\end{defi}
To illustrate how attacker control changes, consider program 
\([\bullet]; \mathit{low} := u < h; \mathit{low'}:=h\) where $u$ is an untrusted variable and $h$ is a secret
trusted variable.
To understand attacker control of this program, we consider 
an initial memory $m(h) = 7$ and attack $ a = u := 5$.
The  low event
$(\mathit{low}, 1)$ in this trace is a release event. The attacker control  
is the set of all attacks that are similar to $a$ and trace $[(u := 5)], (\mathit{low},
1)$ in its influence on knowledge. This corresponds to attacks that set $u$ to values such that $u < 7$. The  assignment to $\mathit{low'}$
changes attacker knowledge as well, but the information that the attacker gets does not depend on the attack:
any trace starting in $m$ and reaching the second assignment produces the low event $(\mathit{low'}, 7)$; 
hence, the attacker control does not change at that event.

Consider the same example but with the two assignments swapped:
\([\bullet]; \mathit{low'}:=h; \mathit{low} := u < h \).
The assignment to $\mathit{low'}$ is a release event that the attacker cannot affect. 
Hence the control includes all attacks that reach this assignment. 
The result of the assignment to $\mathit{low}$ depends on $u$. However, this result does not
change attacker knowledge. Indeed, in this program, the second assignment is not a release
event at all. Therefore, the attacker control is simply all attacks that reach 
the first assignment.

\paragraph{Progress-insensitive control} For progress-insensitive security, attacker control
is defined similarly using the progress-insensitive comparison of attacks.
\begin{defi}[Attacker control (progress-insensitive)]
\[
R_\to(c[\vbullet], m, \vec a, \vec t) \defn  
\{ \vec b
\ |\ \exists\vec q~.~(\vec a, \vec t) \sim^{c[\vbullet], m}_\to (\vec b, \vec q))
\}
\]
\end{defi}
Consider program 
\( [\bullet]; \while{u < h}{\Skip}; \mathit{low}:=1 \). Here, any attack produces a trace that
preserves progress-insensitive noninterference. If the loop is taken, the program produces no low
events, hence, it gives no new knowledge to the attacker. If the loop is not taken, and the low
assignment is reached, this assignment preserves attacker knowledge in a progress-insensitive
way. Therefore, the attacker control is all attacks.

\draftnote{Because attacker control is based on the definition of knowledge segmentation, control necessary
changes at release points. --- there is a lemma about this... }

\section{Robustness}
\label{sec:robust}
 \paragraph{Release control}
\label{sec:releasecontrol}
This section defines \emph{release control} $R^\triangleright$, which
captures the attacker's
influence on release events. Intuitively, release control expresses the extent to which an attacker can affect the _decision_ to produce some release event.

\begin{defi}[Progress-sensitive release control ]
\label{def:control:release:sensitive}
\begin{align*}
R^\triangleright (c[\vbullet], m, \vec a, \vec t) \defn \, &
\{
\vec b\ |\ \exists \vec q~.~(\vec a, \vec t) \sim^{c[\vbullet], m} (\vec b, \vec q)
~\land \\
&\phantom{\{\vec b\ |\ }(\exists \vec{r'}~.~k(c[\vec b], m_\Public, \vec q_\Public)  \supset  k(c[\vec b], m_\Public, \vec q_\Public \tconcat \vec{r'}_\Public)  \\ 
&\phantom{\{\vec b\ |\ (}\lor k(c[\vec b], m_\Public, \vec q_\Public)  \supset k(c[\vec b], m_\Public, \vec q_\Public \Uparrow) \\ 
&\phantom{\{\vec b\ |\ (}\lor \configTwo{c[\vec b]}{m} \Downarrow   
)
\}
\end{align*}\smallskip
\end{defi}

\noindent The definition for release control is based on the one for attacker control with the three
additional clauses,
explained below. These clauses restrict the set of attacks to those
that either terminate or produce a release event. 
Because the progress-sensitive attacker can also learn new information by observing divergence, the 
definition contains an additional clause (on the third line) that uses divergence knowledge to reflect that.

\begin{figure}
\centering
\subfloat[Release control]{\label{fig:releasecontrol} \includegraphics[trim = 0 10 0 10, clip=true]{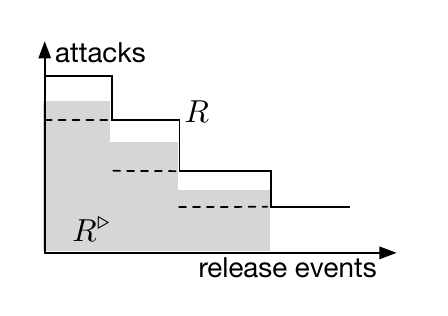} } 
\subfloat[Robustness]{\label{fig:robustness}\includegraphics[trim = 0 10 0 10, clip=true]{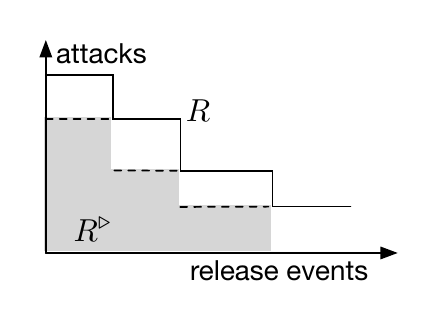}}
\caption{Release control and robustness}
\end{figure}

Figure~\ref{fig:releasecontrol} depicts the relationship between release control and attacker control,
where the $x$-axis corresponds to low events, and the $y$-axis corresponds to attacks. 
 The solid line depicts 
 attacker control $R$, where vertical lines correspond to release events. 
The gray area denotes release control $R^\triangleright$. 
In general, 
for a given attack $\vec a$ and a corresponding trace $\vec t \tconcat \vec r$, where $\vec r$ contains a release
event, we have the following relation between release control and attacker control:
\begin{equation}
 R(c[\vbullet], m, \vec a ,\vec t) \supseteq R^\triangleright (c[\vbullet], m, \vec a, \vec t)
\supseteq R (c[\vbullet], m, \vec a, \vec t \tconcat \vec r ) 
\label{eq:control:relationship}
\end{equation}
Note the white gaps and the gray release control above the dotted
lines on Figure~\ref{fig:releasecontrol}.  The white gaps correspond to difference
$ R(c[\vbullet], m, \vec a ,\vec t) \setminus R^\triangleright (c[\vbullet], m, \vec a, \vec t)$.
This is a set of attacks that do not produce
further release events and that diverge without giving any new information to the attacker---pure availability attacks. The gray zones above the
dotted lines are more interesting.  Every such zone corresponds to the difference $R^\triangleright
(c[\vbullet], m, \vec a, \vec t) \setminus R (c[\vbullet], m, \vec a, \vec t \tconcat \vec r ) $. In
particular, when this set
is non-empty, the attacker can launch attacks corresponding to each of the last three lines of
Definition~\ref{def:control:release:sensitive}:
\begin{enumerate}[(1)]
\item either trigger a different release event  $\vec {r'}$, or 
\item cause program to diverge in a way that also releases information, or 
\item prevent a release event from happening in a way that leads to program 
  termination
\end{enumerate}
Absence of such attacks constitutes the basis for our security conditions in Definitions~\ref{def:robustness:sensitive} and~\ref{def:robustness:insensitive}.  
Before moving on to these definitions, we introduce the progress-insensitive variant of release control.
\draftnote{This definition allows ignoring pure availability attacks such as ($\while{true}{\Skip}$).}

\begin{defi}[Release control (progress-insensitive)]
\begin{multline*}
R^\triangleright_\to
 (c[\vbullet], m, \vec a, \vec t) \defn
\{
\vec b\ |\ \exists \vec q~.~(\vec a, \vec t) \sim^{c[\vbullet], m}_\to (\vec b, \vec q) 
~\land \\
(
\exists \vec{r'}~.~k_\to(c[\vec b], m_\Public, \vec q_\Public)  \supset
 k(c[\vec b], m_\Public, \vec q_\Public \tconcat \vec{r'}_\Public)
\lor\  \configTwo{c[\vec b]}{m} \Downarrow 
)
\}
\end{multline*}\smallskip
\end{defi}

\noindent This definition uses the progress-insensitive variants of similar attacks and release events. It
also does not account for knowledge obtained from divergence.

With the definition of release control at hand we can now define        
semantic conditions for robustness. The intuition is that all attacks   
leading to release events should lead to the same release event.        
Formally, this is defined as inclusion of release control into attacker 
control, where release control is computed on the prefix of the trace   
without a release event.

\begin{defi}[Progress-sensitive robustness]
\label{def:robustness:sensitive}
Program $c[\vbullet]$ satisfies \emph{progress-sen\-sitive robustness} if for all memories $m$, attacks $\vec a$, and traces $\vec t \vec r$,  
such that
$
\configTwo{c[\vec a]}{m} \labarrowstar{}{\vec t} \configTwo{c'}{m'} \labarrowstar{}{\vec r}
$ 
and $\vec r$ contains a release event, i.e.,
$k(c[\vec a], m_\Public, \vec t_\Public) \supset k(c[\vec a], m_\Public, \vec t_\Public \tconcat \vec r_\Public)$,
we have \[R^\triangleright (c[\vbullet], m, \vec a, \vec t) \subseteq
R(c[\vbullet], m, \vec a, \vec t \tconcat \vec r)\]\vspace{-3 pt}
\end{defi}

\noindent Note that because of Equation~\ref{eq:control:relationship}, set inclusion in the above definition could be replaced with strict equality, but we
use $\subseteq$ for compatibility with future definitions. 
Figure~\ref{fig:robustness} illustrates the relation between release control and attacker control
for robust programs. Note how release control is bounded by the attacker control at the next
release event.

\paragraph{Examples} We illustrate the definition of robustness with a few examples.

Consider program $[\bullet]; \mathit{low} := u < h$, and memory such that $m(h) = 7$. 
This program is rejected by Definition~\ref{def:robustness:sensitive}. To see this,  pick an
$a= u:= 5$, and consider the part of the trace preceding the low assignment.  Release
control $R^\triangleright(c[\vbullet], m, a, [(u,5)]) $ is all attacks that reach the assignment to
$\mathit{low}$. On the other hand, the attacker control $R(c[\vbullet],m, a, [(u,5)] \tconcat (\mathit{low},
1))$ is the set of all attacks where $u<7$, which is smaller than $R^\triangleright$. Therefore this
program does not satisfy the condition. 

Program $[\bullet]; low := h; \mathit{low'} := u < h$ satisfies robustness. The only release event
here corresponds to the first assignment. However, because the knowledge given by that assignment
does not depend on untrusted variables, the release control includes all attacks that reach the
assignment. 

Program $[\bullet]; \ifThenelse{u>0}{\mathit{low} := h}{\Skip}$ is rejected. Consider memory
$m(h)=7$, and attack $a = u := 1$ that leads to low trace $[(u,1)] \tconcat (\mathit{low}, 7)$. The attacker
control for this attack and trace is the set of all attacks such that $u > 0$. On the other hand,
release control $ R^\triangleright (c[\vbullet], m, \vec a, [(u,1)])$ is the set of all attacks that lead to
termination, which includes attacks such that $u \leq 0$. Therefore, the release control 
corresponds to a bigger set than the attacker control.

\draftnote{
Note that program
$[\bullet]; \\ \ifThenelse{u>0}{\mathit{low}:=h}{\Skip}; \\\while{\cod{true}}{\Skip}$. is accepted by this definition. It is however, semantically equivalent to program 
$  [\bullet]; \\\while{u>0}{\Skip}; \\ \mathit{low}:=h;\\ \while{\cod{true}}{\Skip}$
}

Program $[\bullet]; \while{u>0}{\Skip}; \mathit{low}:=h$ is accepted. Depending on the attacker-controlled variable the release event is reached. However, this is an example of availability
attack, which is ignored by Definition~\ref{def:robustness:sensitive}.

Program $[\bullet]; \while{u>h}{\Skip}; \mathit{low}:=1$ is rejected.
Any attack leading to the low
assignment restricts the control to attacks such that $u \leq h$.  However, release
control includes attacks $u > h$, because the attacker learns
information from divergence.

The definition of progress-insensitive robustness is similar to
Definition~\ref{def:robustness:sensitive}, but uses progress-insensitive variants of release events,
control, and release control. As a result,
program $[\bullet]; \while{u>h}{\Skip}; \mathit{low}:=1$ is accepted: attacker control
is all attacks.

\begin{defi}[Progress-insensitive robustness]
\label{def:robustness:insensitive}
Program $c[\vbullet]$ satisfies \emph{progress-insen\-sitive robustness} if for all memories $m$, attacks $\vec a$, 
and traces $\vec t \vec r$, such that
$
\configTwo{c[\vec a]}{m} \labarrowstar{}{\vec t} \configTwo{c'}{m'} \labarrowstar{}{\vec r}
$ 
and $\vec r$ contains a release event, i.e.,
$k_\to(c[\vec a], m_\Public, \vec t_\Public) \supset k(c[\vec a], m_\Public, \vec t_\Public \tconcat \vec r_\Public)$,
we have \[R^\triangleright_\to(c[\vbullet], m, \vec a, \vec t) \subseteq R_\to(c[\vbullet], m, \vec a, \vec t \tconcat \vec r)\]
\end{defi}

\draftnote{Do we need  more examples here?}

\section{Endorsement}
\label{sec:endorsement}

This section extends the semantic policies for robustness in a way that allows \emph{endorsing} 
attacker-provided values. 

\begin{revisit}
\end{revisit}

\draftnote{Endorsement: more introduction}

\paragraph{Syntax and semantics}
We add endorsement to the language:
\begin{align*}
c[\vbullet]::=\; & \dots \ |\ x:=\EndorseLab{\eta}{e}
\end{align*}
We assume that every endorsement in the program source has  a unique 
\emph{endorsement label} $\eta$.
Semantically, endorsements produce 
\emph{endorsement events}, denoted by$\mathit{endorse} (\eta, v)$, which record the label of the endorsement statement $\eta$ 
together with the value $v$ that is endorsed. 
\begin{mathpar}
\inferrule{
  \configTwo{e}{m} \downarrow {v}
}{\configTwo{x:=\EndorseLab{\eta}{e}}{m} \labarrow{}{\mathit{endorse}(\eta, v)} \configTwo{\halt}{m[x \mapsto v]}}
\end{mathpar}
Whenever the endorsement label is unimportant, we omit it from the examples. Note that $\mathit{endorse}(\eta, v)$ events need not mention variable name $x$ since that information
is implied by the unique label $\eta$.

Consider example program $[\bullet] ;\mathit{low} := \EndorseLab{\eta_1}{u < h}$. This
program does not satisfy Definition~\ref{def:robustness:sensitive}. The reasoning for this is
exactly the same as for program
$[\bullet] ;\mathit{low} := {u < h}$ from Section~\ref{sec:robust}. 

\paragraph{Irrelevant attacks} 
Endorsement of certain values gives attacker some control over the knowledge. The key technical element of this section is the notion of _irrelevant attacks_, which defines the set of attacks that are endorsed, and that are therefore _excluded_ when comparing attacker control with release control. 
We define irrelevant attacks formally below, based on the trace that is produced by a program.

Given a program $c[\bullet]$, starting memory $m$, and a trace $\vec
t$, irrelevant attacks, denoted here by 
$\Corrupt{c[\vbullet], m, \vec t}$, are the attacks that lead to the same sequence of endorsement events as in $\vec t$, until they necessarily disagree on one of the endorsements.
  Because the influence of these attacks
is reflected at endorsement events, we exclude them from consideration when comparing with
attacker control. 

We start by defining \emph{irrelevant traces}. Given a trace $\vec t$,
 irrelevant traces for $\vec t$ are all traces $\vec{t'}$ that agree with $\vec t$
on some prefix of endorsement events until they necessarily disagree
on some endorsement.  We define this set as follows.

\draftnote{bind $i$}

\begin{defi}[Irrelevant traces]
\label{def:irrelevant:traces}
  Given a trace $\vec t$, where endorsements are marked as
  $\mathit{endorse}(\eta_j, v_j)$,  define a set of irrelevant traces
  based on the number of  endorsements in $\vec t$~as~$\idxcorrupt{i}{\vec t}$:
 $\idxcorrupt{0}{\vec t} = \emptyset$, and
\begin{align*}
 \idxcorrupt{i}{\vec t} = \{ & \vec{t'}\ |\ \vec{t'} = \vec{q} \tconcat 
 \mathit{endorse}(\eta_{i}, v'_i)  \tconcat \vec{q'}   \} \mbox{\ such that  } \\
& \mbox{$\vec{q}$ is a prefix of $\vec{t'}$ with $i-1$  events all of  which agree with $\mathit{endorse}$ events in $\vec t$, and} \\
& \mbox{$v_i \neq v'_i$ } 
\end{align*}
Define \(
\corrupt{\vec t} \defn \bigcup_i  \idxcorrupt{i}{\vec t}  \) as a set of \emph{irrelevant traces} w.r.t. $\vec t$.
\end{defi}

With the definition of irrelevant traces at hand, we can define irrelevant attacks:
irrelevant attacks are  attacks that  lead to irrelevant traces. 
\begin{defi}[Irrelevant attacks]
\label{def:irrelevant:attacks}
Given a program $c[\vbullet]$, initial memory $m$, and a trace $\vec t$, such that 
$\configTwo{c[\vbullet]}{m} \labarrowstar{}{\vec t}$, define \emph{irrelevant attacks} $\Corrupt{c[\vbullet], m, \vec t} $ as
\[
\Corrupt{c[\vbullet], m, \vec t} 
\defn
\{ \vec{a}\ |\ 
\configTwo{c[\vec a]}{m} \labarrowstar{}{\vec t'} \land
\vec{t'} \in \corrupt{\vec t}
\}
\]
\end{defi}

\paragraph{Security}
The security conditions for robustness can now be adjusted to accommodate endorsements that happen
along traces. The idea is to exclude irrelevant attacks from the left-hand side of
Definitions~\ref{def:robustness:sensitive}
and~\ref{def:robustness:insensitive}. This security condition, which has
both progress-sensitive and progress-insensitive versions, expresses
roughly the same idea as _qualified robustness_~\cite{Myers:Sabelfeld:Zdancewic:JCS06},
but in a more natural and direct way.

\begin{defi}[Progress-sensitive robustness with endorsements]
\label{def:robustness:qualified:sensitive}
Program $c[\vbullet]$ satisfies \emph{progress-sensitive robustness with endorsement}
if for all memories $m$, attacks $\vec a$, and traces $\vec t \vec r$,
such that
$
\configTwo{c[\vec a]}{m} \labarrowstar{}{\vec t} \configTwo{c'}{m'} \labarrowstar{}{\vec r}
$ 
and $\vec r$ contains a release event, i.e., 
$k(c[\vec a], m_\Public, \vec t_\Public) \supset k(c[\vec a], m_\Public, \vec t_\Public \tconcat \vec r_\Public)$,
we have \[R^\triangleright (c[\vbullet], m, \vec a, \vec t) \setminus \Corrupt{c[\vbullet], m, \vec t \tconcat \vec r} \subseteq  R(c[\vbullet], m, \vec a, \vec t \tconcat \vec r)\]

\end{defi}
\begin{figure}
\noindent
\centering
\subfloat[Irrelevant attacks]{\parbox{5cm}{\centering
\label{fig:endorsements:irrelevant}\includegraphics[scale=1.2]{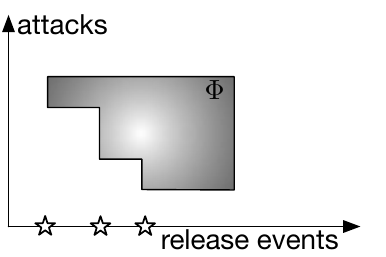}}}
\subfloat[Robustness w/o endorsements (unsatisfied)]{\parbox{5.6cm}{\centering
\label{fig:endorsements:nonrobust}\includegraphics[scale=1.2]{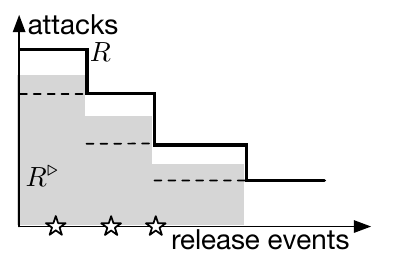}}}
\subfloat[Robustness with endorsements
(satisfied)]{\parbox{5cm}{\centering
\label{fig:endorsements:robust}\includegraphics[scale=1.2]{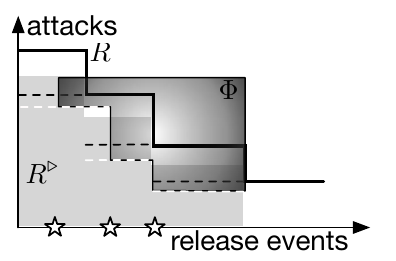}}}

\caption{Irrelevant attacks and robustness with endorsements}
\end{figure}\smallskip

\noindent We refer to the set $R^\triangleright (c[\vbullet], m, \vec a, \vec t)
\setminus \Corrupt{c[\vbullet], m, \vec t \tconcat \vec r} $ as a set of
\emph{relevant attacks}. 
 Figures~\ref{fig:endorsements:irrelevant} to~\ref{fig:endorsements:robust} visualize irrelevant attacks
and the semantic condition
of Definition~\ref{def:robustness:qualified:sensitive}.
Figure~\ref{fig:endorsements:irrelevant} shows the set of irrelevant
attacks, depicted by the shaded gray area. This set increases at endorsement events marked by stars.
Figure~\ref{fig:endorsements:nonrobust} shows an example trace where robustness is not satisfied --- 
the gray area corresponding to release control $R^\triangleright$ exceeds the attacker control
(depicted by the solid line).  Finally, in
Figure~\ref{fig:endorsements:robust}, we superimpose
Figures~\ref{fig:endorsements:irrelevant}
and~\ref{fig:endorsements:nonrobust}. This illustrates that when
the set of irrelevant attacks is excluded from the release control (the area under white dashed
lines), the program is accepted by robustness with endorsements.
\paragraph{Examples}

Program $[\bullet] ;\mathit{low} := \EndorseLab{\eta_1}{u < h}$ is accepted by Definition~\ref{def:robustness:qualified:sensitive}. Consider initial memory $m(h) = 7$, and an
attack $u:=1$; this produces a trace $[(u,1)] \mathit{endorse}(\eta_1,1)$. The endorsed assignment also produces a release event. We have that 
\begin{iteMize}{$\bullet$}
\item Release control $R^\triangleright$ is the 
set of all attacks that reach the low assignment. 

\item 
Irrelevant traces $\corrupt{[(u,1)] \mathit{endorse}(\eta_1,1)}$ is 
a set of traces that end in endorsement event $\mathit{endorse}{(\eta_1, v)}$ such that $v \neq 1$. Thus, 
irrelevant attacks $\Corrupt{[\bullet]; \mathit{low} := \EndorseLab{\eta_1}{u < h}, m, [(u,1)] \mathit{endorse}(\eta_1,1)}$
must consist of attacks that reach the low assignment and set $u$ to values  $u \geq 7$.

\item The left-hand side of Definition~\ref{def:robustness:qualified:sensitive} is therefore the set of attacks that reach the endorsement and set $u$ to $ u < 7$. 

\item As for the attacker control on the right-hand side, it consists of attacks that
  set $u < 7$. Hence, the set inclusion of Definition~\ref{def:robustness:qualified:sensitive} holds and the
  program is accepted.
\end{iteMize}\smallskip

\noindent Program $[\bullet]; \mathit{low} := \EndorseLab{\eta_1}{u}; \mathit{low'}:= {u < h''}$ is accepted. The 
endorsement in the first assignment implies that all relevant attacks must agree on the value of $u$, and, consequently, they agree  on the value of $u < h''$, which gets assigned to $\mathit{low'}$. This also
 means that relevant attacks belong to the attacker control (which contains all attacks 
that  agree on $u < h''$). 

Program $[\bullet]; \mathit{low} := \EndorseLab{\eta_1}{u < h}; \mathit{low'}:= {u < h''}$ is
rejected. Take initial memory such that $m(h) \neq m(h')$. The set of relevant attacks after the second assignment contains
attacks that agree on $u < h$ (due to the endorsement), but not necessarily on $u < h''$. The latter, however, is the requirement 
for the attacks that belong to the attacker control.

Program $[\bullet]; \ifThenelse{u>0}{\mathit{h'}:=\Endorse{u}}{\Skip}; \mathit{low} := h' < h$ is
rejected.  Assume initial memory where $m(h) = m(h') = 7$. Consider
attack $a_1$ that sets $u:= 1$ and consider the trace $\vec t_1$ that it gives. This trace endorses $u$ in the $\cod{then}$ branch, overwrites the
value of $h'$ with $1$, and produces a release event $(\mathit{low},
1)$.  Consider another attack $a_2$ that sets
$u:=0$, and consider the corresponding trace $\vec t_2$. This trace contains release event $(\mathit{low}, 0)$ without any
endorsements. Now, attacker control $R(c[\vbullet], m, a_2, \vec t_2)$ excludes $a_1$, because of
the disagreement at the release event.  At the same time, $a_1$ is a relevant attack for $a_2$,
because no endorsements happen along~$\vec t_2$. 

Consider program $c[\vbullet]$, which contains no endorsements. In this case, for all possible traces
$\vec t$, we have that $\corrupt{\vec t} = \idxcorrupt{0}{\vec t} = \emptyset$. Therefore, by
Definition~\ref{def:irrelevant:attacks} it must be that $\Corrupt{c[\vbullet], m, \vec t} =
\emptyset$ for all memories $m$ and traces $\vec t$.  This indicates that for programs without
endorsements, progress-sensitive robustness with endorsements
(Definition~\ref{def:robustness:qualified:sensitive}) conservatively reduces to the earlier
definition of progress-sensitive robustness (Definition~\ref{def:robustness:sensitive}).

Progress-insensitive robustness with endorsement is defined similarly. The intuition for the definition 
remains the same, while  we use progress-insensitive variants of progress control and control:

\begin{defi}[Progress-insensitive robustness with endorsement]
\label{def:robustness:qualified:insensitive}
Program $c[\vbullet]$ satisfies \emph{progress-insensitive robustness
with endorsement} if for all memories $m$, attacks $\vec a$, and traces $\vec t \tconcat \vec r$, 
such that
$
\configTwo{c[\vec a]}{m} \labarrowstar{}{\vec t} \configTwo{c'}{m'} \labarrowstar{}{\vec r}, 
$ 
and $\vec r$ contains a release event, i.e., 
$k_\to(c[\vec a], m_\Public, \vec t_\Public) \supset k(c[\vec a], m_\Public, \vec t_\Public \tconcat \vec r_\Public)$,
we have \[R^\triangleright_\to(c[\vbullet], m, \vec a, \vec t) \setminus \Corrupt{c[\vbullet], m, \vec t \tconcat \vec r} \subseteq R_\to(c[\vbullet], m, \vec a, \vec t \tconcat \vec r)\]

\end{defi}

As a final note in this section, observe that because of the particular use of irrelevant attacks in Definitions~\ref{def:robustness:qualified:sensitive} and~\ref{def:robustness:qualified:insensitive} it is sufficient for us to define irrelevant traces so that they only match at the endorsement events. A slightly more generalized notion of irrelevance
would require $\vec q$ in Definition~\ref{def:irrelevant:traces} to be similar to a prefix of $\vec{t'}$. 

\section{Enforcement}
\label{sec:enforcement}

We now explore how to enforce
robustness using a security type
system. While this section focuses on progress-insensitive enforcement, it is possible to refine the type system to deal with
progress sensitivity (modulo availability attacks)~\cite{Volpano:Smith:Probabilistic:CSFW,ONeil+:CSFW06}.  %
Figures~\ref{fig:types:expressions} and~\ref{fig:types:commands} display typing rules for
expressions and commands. This type system is based on the one of~\cite{Myers:Sabelfeld:Zdancewic:JCS06}
and is similar to many standard security type systems.

\paragraph{Declassification} We extend the language with a language construct for
\emph{declassification} of expressions $\decl e$. Whereas in earlier
examples, we considered an assignment $l := h$ to be secure if it did
not violate robustness, we now require information flows from
public to secret to be mediated by declassification.
We note that declassification has no
additional semantics and, in the context of our simple language, can be inferred automatically.
This may be achieved by placing declassifications in  public assignments that appear in  
trusted code, i.e., in  non-$\bullet$ parts of the program.
Moreover, making declassification explicit has the following motivations:

\begin{enumerate}[(1)]
\item On the enforcement level, the type system conveniently ensures that a non-progress release event may happen only at declassification. 
  All other assignments preserve progress-insensitive knowledge.
\item Much of the related work on language-based declassification policies uses similar type systems.
  Showing our security policies can be enforced using such systems
  makes the results more general.
\end{enumerate}

\paragraph{Typing of expressions}
Type rules for expressions have form $\G \vdash e : \level, D$ 
where $\level$ is the level of the expression, and $D$ is a set of variables that may be declassified. The declassification is the most interesting rule among expressions. 
It  downgrades the confidentiality level of the expression by returning $\ell \sqcap (\Public, \Untrusted)$, and counts all variables in $e$ as declassified.

\vspace{-2em}
\begin{figure}
\begin{mathpar}
\inferrule{}{\G \vdash n : \level, \emptyset } 
\and
\inferrule{}{\G \vdash x : \G(x), \emptyset }
\and 
\inferrule{
\G \vdash e_1 : \level_1, D_1 \and
\G \vdash e_2 : \level_2, D_2 
}{\G \vdash e_1\ \mathit{op}\ e_2 : \level_1 \join \level_2 , D_1 \cup D_2} 
\and
\inferrule[(T-DECL)]{
\G \vdash e : \level, D 
}{
\G \vdash \decl{e} :  \ell \sqcap (\Low, \Untrusted) , \mathit{vars} (e) 
}
\and
\end{mathpar}
\caption{Type system: expressions}
\label{fig:types:expressions}
\begin{mathpar}
\inferrule[(T-SKIP)]{}{\G, \pc \vdash \Skip } 
\and
\inferrule[(T-SEQ)]{\G, \pc \vdash c_1 \and \G, \pc \vdash c_2}{\G, \pc \vdash c_1; c_2}
\and
\inferrule[(T-ASGMT)]{\G  \vdash e : \level, D
\and
\level \join \pc \sqsubseteq \G(x)
\and
\forall y \in D~.~ \G(y) \sqsubseteq (\High, \Trusted)
\and
D \neq \emptyset \implies \pc \sqsubseteq (\Low, \Trusted )
  }{\G, \pc \vdash x := e}
\and
\inferrule[(T-IF)]{
\G \vdash e: \level, \emptyset
\and
\G, \pc \join \level \vdash c_1 
\and
\G, \pc \join \level \vdash c_2
}{
\G, \pc  \vdash \ifThenelse{e}{c_1}{c_2}
}
\and
\inferrule[(T-WHILE)]{
\G \vdash e: \level, \emptyset 
\and
\G, \pc \join \level \vdash c
}{
\G, \pc \vdash \while{e}{c}
}
\and
\inferrule[(T-HOLE)]{
\pc \sqsubseteq (\Low, \Untrusted) 
}{
\G, \pc \vdash \bullet
}
\and
\inferrule[(T-ENDORSE)]{
  \pc \join \G(x) \sqsubseteq (\Secret, \Trusted) \and
  \pc \sqsubseteq \G(x) \and
  \G \vdash e : \ell, \emptyset \and
  \ell \sqcap (\Secret, \Trusted) \sqsubseteq \G(x)
}{
\G, \pc \vdash x := \Endorse{e}
}
\end{mathpar}
\caption{Type system: commands}
\label{fig:types:commands}
\end{figure}

\paragraph{Typing of commands}
The typing judgments for commands have the 
form $\Gamma, \pc \vdash c$. The rules are standard for a security type system. We highlight
 typing of assignments, endorsement, and holes.

Assignments have two extra clauses for when the assigned expression contains a
declassification $(D \neq \emptyset)$. 
The rule (T-ASGMT) requires all variables that
can be declassified have high integrity. The rule also bounds the $\pc$-label by $ (\Public,
\Trusted)$, which enforces that no declassification happens in untrusted or secret
contexts. These requirements guarantee that the information released by the declassification does not
directly depend on the attacker-controlled variables. 

The typing rule for endorsement (T-ENDORSE) requires that the $\pc$-label is trusted and that the result of the endorsement is stored
in a trusted variable:  $\pc \join \G(x) \sqsubseteq (\Secret,
\Trusted)$.  Note that requiring a trusted $\pc$-label  is crucial,
while the restriction that $x$ is trusted could easily be lifted, since trusted values may flow into untrusted variables.
Because endorsed expressions preserve their
confidentiality level, we also check that $x$ has the right security level to store the result of the expression. This
is done by demanding that $\level \sqcap (\Secret, \Trusted) \sqsubseteq \G(x)$, where taking meet of $\ell$ and $(\Secret, \Trusted)$ 
boosts integrity, but keeps the confidentiality level of $\level$.

The rule for holes forbids placing attacker-provided code in high confidentiality contexts. 
For simplicity, we disallow declassification in the guards of $\cod{if}$ 
and $\cod{while}$.

\subsection{Soundness} This section shows that the type system of Figures~\ref{fig:types:expressions}
and~\ref{fig:types:commands} is sound. We formulate top-level soundness in Proposition~\ref{prop:control:soundness}. 
The proof of Proposition~\ref{prop:control:soundness} appears in the end of the section.

 \begin{prop}
\label{prop:control:soundness}
 If $\Gamma, \pc \vdash c[\vbullet]$ then for all attacks $\vec a$, memories $m$, and traces $\vec t \tconcat \vec r$ produced 
 by $\configTwo{c[\vec a]}{m}$, where $k_\to(c[\vec a], m_\Public, \vec t_\Public) \supset k(c[\vec a], m_\Public, \vec t_\Public \tconcat \vec r_\Public)$,
we have that
\[R^\triangleright_\to(c[\vbullet], m, \vec a, \vec t) \setminus \Corrupt{c[\vbullet], m, \vec t \tconcat \vec r} \subseteq R_\to(c[\vbullet], m, \vec a, \vec t \tconcat \vec r)\]
 \end{prop}

\paragraph{Auxiliary definitions}
We introduce an auxiliary definition of progress-insensitive noninterference along a part of a
trace, abbreviated PINI, which we will use in the proof of
Proposition~\ref{prop:control:soundness}. 
 Figure~\ref{fig:proof:soundness} shows the high-level
structure of the proof. 
We define \emph{declassification
  events} to be low events that involve
declassifications.
The central property of this proof --- the control backbone lemma
(Lemma~\ref{lemma:control:backbone}) --- captures the behavior of similar attacks and traces that
are generated by well-typed commands. Together with the Advancement
Lemma, it shows that declassification events soundly approximate release events.
 The proof of Proposition~\ref{prop:control:soundness} follows
directly from the Control Backbone and Advancement lemmas.

 \begin{defi}[Progress-insensitive noninterference along a part of a trace]
   Given a program $c$, memory $m$, and two traces $\vec t$ and $\vec t^+$ such that $\vec t^+ $ is
   an extension of $\vec t$, we say that $c$ satisfies \emph{progress-insensitive noninterference}
   along the part of the trace from $\vec t$ to $\vec t^+$, denoted by $\mathrm{PINI}(c, m, \vec t,
   \vec t^+)$ whenever for the low events in the corresponding traces
  $\vec \ell_n \defn \vec t_\Public$ and $ \vec \ell_N \defn  \vec t^+_\Public, 
   n  \leq N$, it holds that
\[
\forall i~.~ n < i < N~.~ k_\to(c,m_\Public, \vec \ell_i) \subseteq k (c, m_\Public, \vec \ell_{i+1})
\]

 \end{defi}

\begin{figure}

\centering
\includegraphics{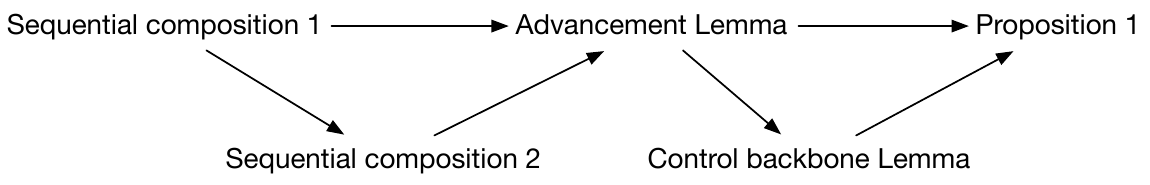}

\caption{High-level structure of proof of Proposition~\ref{prop:control:soundness}}
\label{fig:proof:soundness} 
\end{figure}

 \begin{lem}[Noninterference for no declassifications] 
 \label{lemma:nodeclass:ni}
 Given a program $c$ without declassifications such that $\Gamma, \pc \vdash c$ then for all
 memories $m$ and possible low events $\vec \ell \tconcat \ell'$ such that 
\[ \configTwo{c}{m}
 \labarrowstar{}{\vec \ell}  \configTwo{c'}{m'} 
 \labarrowstar{}{\ell'}  \configTwo{c''}{m''} \]
it holds that $k_\to(c, m, \vec \ell) \subseteq k(c, m, \vec \ell \tconcat \ell')$.
 \end{lem}
 \begin{proof}
 By induction on $c$ (cf. \cite{Askarov:Sabelfeld:SP07}). 
 \end{proof}

 \begin{lem}[Noninterference for the tail of sequential composition]
 \label{lemma:knowledge:tail}
 Assume a program $c$ such that for all memories $m$ and low events $\ell$, 
 such that 
$\configTwo{c}{m} \labarrowstar{}{\ell} \configTwo{c'}{m'}$, it holds that
 $k_\to(c,m_\Low, \epsilon) \subseteq k(c, m_\Low, \ell)$.  
Then for all
 programs $c_0$, initial memories $i$, and low events $\vec l_0$, such that 
 \[
 \configTwo{c_0; c}{i} \labarrowstar{}{\vec \ell_0} \configTwo{c}{i'} \labarrowstar{}{\ell'}
 \]
 we have $k_\to(c_0; c, i_\Low, \vec \ell_0) \subseteq k(c_0; c, i_\Low, \vec \ell_0 \tconcat \ell')$.
 \end{lem}
 \begin{proof}
 Assume the set inclusion of the lemma's statement does not
 hold. By Definition~\ref{def:knowledge:progress}, there must  exist an initial memory $m$, such that  $m
 =_\Public i$ and $\configTwo{c_0;c}{m} \labarrowstar{}{\vec \ell_0}
 \configTwo{c}{m'} \labarrowstar{}{\ell''}$, but $ \ell' \neq
 \ell''$. Because $m =_\Public i$ and both traces produce $\vec \ell_0$, it
must also be that $m' =_\Public i'$. But this also implies that
$m' \not\in k (c, i'_\Public, \ell')$, that is, 
$k_\to (c, i'_\Public, \epsilon) \not\subseteq k(c,i'_\Public,
\ell')$, which contradicts the main assumption about $c$.  
\end{proof}

The following two helper lemmas correspond to the sequential composition sub-cases of the
Advancement Lemma.
Lemma~\ref{lemma:control:seq1} captures the special case when the first command in the sequential composition $c_1[\bullet]; c_2[\bullet]$ does not produce a declassification event, while
Lemma~\ref{lemma:control:seq2} considers the general case when a declassification event may be produced by either of $c_1[\bullet]$ or $c_2[\bullet]$.

\begin{lem}[Sequential composition 1]
 \label{lemma:control:seq1}

 Given 
 \begin{iteMize}{$\bullet$}
 \item program $\vec c_0[\vbullet]$ such that $ \G, \pc \vdash c_0[\vbullet]$,
 \item initial memory $m_0$,
 \item two initial attacks $\vec a_0, \vec b_0$, 
 \item two intermediate configurations $\configTwo{c_1[\vec a_1]; c_2[\vec a_2]}{m}$ and 
 $\configTwo{c_1[\vec b_1]; c_2[\vec b_2]}{s}$ such that 
 \item 
   $\configTwo{c_0[\vec a_0]}{m_0} \labarrowstar{}{\vec t'} 
   \configTwo{c_1[\vec a_1]; c_2[\vec a_2]}{m}  \labarrowstar{}{\vec t_\alpha} \configTwo{c_2[\vec a_2]}{m'}
   \labarrowstar{}{\vec t_\beta \tconcat r}  $ 
 \item 
   $\configTwo{c_0[\vec b_0]}{m_0} \labarrowstar{}{\vec{q'}} 
   \configTwo{c_1[\vec b_1]; c_2[\vec b_2]}{s}  \labarrowstar{}{\vec{q''} \tconcat r'}$
 \item 
 $\mathrm{PINI}(c_0[\vec a_0], m_0, \vec t', \vec{t'} \tconcat \vec t_\alpha \tconcat \vec t_\beta)$
 \item
 $\mathrm{PINI}(c_0[\vec b_0], m_0, \vec{q'}, \vec{q'} \tconcat \vec{q''})$
 \item $r$ and $r'$ are declassification events 
 \item 
 $\vec b_0 \not\in \Corrupt{ c_0[\vbullet], m_0, \vec{t'} \tconcat \vec t_\alpha \tconcat \vec t_\beta \tconcat r}$
 \item $\vec t'_\star = \vec{q'}_\star$
 \item
 $m =_\Trusted s$
 \end{iteMize}
 then $\vec{q''} = \vec q_\alpha \tconcat \vec q_\beta$ such that 
 \begin{iteMize}{$\bullet$}
 \item 
 $  \configTwo{c_1[\vec a_1]; c_2[\vec a_2]}{s}  \labarrowstar{}{\vec q_\alpha} \configTwo{c_2[\vec a_2]}{s'} 
 \labarrowstar{}{\vec q_\beta \tconcat r}
 $

 \item
 $\vec {t_\alpha}_\star = \vec {q_\alpha}_\star$
 \item 
 $m' =_\Trusted s'$
 \end{iteMize}

 \end{lem}

 \begin{proof}
 By induction on the structure of $c_1[\vbullet]$. Case $\Skip$ is immediate. Consider the other cases.

 \begin{iteMize}{$\bullet$}

 \item case $[\vbullet]$

 In this case $\vec a_1 = a_1$ and $\vec b_1 = b_1$.
By assumption, $a_1$ and $b_1$ are fair attacks, which means that $\vec t_\alpha$ has no release events and no
assignments to trusted variables. Similarly, because no low assignments can be produced when running $\vec b_1$, then by 
Definition~\ref{def:fairattack} there must be $s'$ and $\vec q_\alpha$ that would satisfy the demand of the lemma.

 \item case $x := e $

We consider confidentiality and integrity properties separately.
 \begin{desCription}
 \item\noindent{\hskip-12 pt\bf Confidentiality:}\
   We show that even if a low event is possible, it is not a release event. We have two
   cases, based on the confidentiality level of $x$.
 \begin{enumerate}[(a)]
 \item
 $\G(x) = (\Public, \_)$

 A low event is generated by the low assignment. By Lemma~\ref{lemma:nodeclass:ni} and Lemma~\ref{lemma:knowledge:tail}
  the assignment must not be a release event. 

 \item
 $\G(x) = (\Secret, \_)$

 In this case no low events are generated.
 \end{enumerate}
 \item\noindent{\hskip-12 pt\bf Integrity:}\
 Next, we show that the resulting memories agree on trusted values. The two cases are 
 \begin{enumerate}[(a)]
 \item
 $\G(x) = (\_, \Trusted)$
 In this case it must be that $\G (e) = (\_, \Trusted)$ and, hence, $m(e) = s(e)$. Therefore $m' =_\Trusted s'$. 
 \item

 $\G(x) = (\_, \Untrusted)$ 
 Assignment to $x$ does not change how memories agree on trusted values.
 \end{enumerate}

 \end{desCription}
 \item case $x := \EndorseLab{\eta}{e}$

 We consider the confidentiality and integrity properties of this command separately. 

 \begin{desCription}
 \item\noindent{\hskip-12 pt\bf Confidentiality:}\ Similar to the case
 for assignment. 
 \item\noindent{\hskip-12 pt\bf Integrity:}\ We consider two cases. 
 \begin{enumerate}[(a)]
 \item $\G(x) = (\_, \Trusted)$

   In this case, the trace produces an event $\mathit{endorse} (\eta,
   v)$. We note $\vec b_0 \not\in
   \Corrupt{c_0[\vbullet], m_0, \vec{t'} \tconcat \vec t_\alpha \tconcat \vec t_\beta \tconcat r}$.  In particular, we have that 
$\vec{q'} \tconcat \vec{q''} \tconcat r' \not\in \corrupt{\vec{t'} \tconcat \vec t_\alpha \tconcat \vec t_\beta \tconcat r}$. If we assume 
that the current command is the $i$-th endorsement in the trace, we have that 
$\vec{q'} \tconcat \vec{q''} \tconcat r' \not\in \idxcorrupt{i}{\vec{t'} \tconcat \vec t_\alpha \tconcat \vec t_\beta \tconcat r}$. 
But we also know that $\vec{t'}_\star =_\Trusted \vec{q'}_\star$. Because, by the rule (T-ENDORSE), 
the result of endorsement is assigned to trusted variables, 
this implies that both $\vec{q'}$ and $\vec{t'}$ must agree on the endorsed
   values. Therefore, the only possibility with which $\vec{q'} \tconcat \vec{q''} \tconcat r' \not\in \idxcorrupt{i}{\vec t}$ 
is that
   $\vec{q''}$ generates $\mathit{endorse}(\eta, v)$ as well. This implies that value $v$ is assigned to $x$ in both cases, which 
   guarantees that $m' =_\Trusted s'$.

 \item $\G(x) = (\_, \Untrusted)$ 
Not applicable by (T-ENDORSE).

 \end{enumerate}
 \end{desCription}

 \item case $c_\alpha; c_\beta$
 \

 By two applications of induction hypothesis: one to $c_\alpha; (c_\beta; c_2[\vbullet])$ and the other one 
 to $c_\beta; c_2[\vbullet]$.
 \item case $\ifThenelse{e}{c_{\cod{true}}}{c_{\cod{false}}}$

 We have the following cases based on the type of expression $e$. 

 \begin{enumerate}
 \item $\G(e) = (\_, \Trusted)$

 In this case both branches are taking the same branch and we are done by induction hypothesis. 

 \item $\G(e) = (\_, \Untrusted)$

   In this case neither of $c_\cod{true}$ or $c_\cod{false}$ contain declassifications or high
   integrity assignments. This guarantees that $m'=_\Trusted s'$. 

 \end{enumerate}

 \item case $\while{e}{c_{\cod{loop}}}$

 Similar to sequential composition and conditionals.

 \end{iteMize}

 \end{proof}

 \begin{lem}[Sequential composition 2]
 \label{lemma:control:seq2}
 Given 
 \begin{iteMize}{$\bullet$}
 \item program $c_0[\vbullet]$ such that $\G, \pc \vdash c_0[\vbullet]$
 \item initial memory $m_0$
 \item two initial attacks $\vec a_0, \vec b_0$ 
 \item two intermediate configurations $\configTwo{c_1[\vec a_1]; c_2[\vec a_2]}{m} $
 and $\configTwo{c_1[\vec b_1]; c_2[\vec b_2]}{s}$ such that 

 \item 
 $\configTwo{c_0[\vec a_0]}{m_0} \labarrowstar{}{\vec{t'}}
 \configTwo{c_1[\vec a_1]; c_2 [\vec a_2]}{m}
 \labarrowstar{}{\vec{t''}}
 \configTwo{c'_1[\vec a'_1]; c_2[\vec a_2]}{m'}
 \labarrow{}{(x,v)}
 \configTwo{c''_1[\vec a''_1]; c_2[\vec a_2]}{m''}
 $

 \item 
 $\configTwo{c_0[\vec b_0]}{m_0} \labarrowstar{}{\vec{q'}}
 \configTwo{c_1[\vec b_1]; c_2[\vec b_2]}{s}
 \labarrowstar{}{\vec{q''}}
 \configTwo{d'}{s'}
 \labarrow{}{ (y,u)}
 \configTwo{d''}{s''}
 $

 \item $\mathrm{PINI}(c_0[\vec a_0], m_0, \vec{t'}, \vec{t'} \tconcat  \vec{t''})$

 \item $\mathrm{PINI}(c_0[\vec b_0], m_0, \vec{q'}, \vec{q'} \tconcat \vec{q''})$

 \item $(x,v)$ and $(y,u)$ are declassification events

 \item $\vec b_0 \not\in \Corrupt{c_0[\vbullet], m_0, \vec{t'} \tconcat \vec{t''} \tconcat (x,v) }$

 \item ${\vec{t'}}_\star   = {\vec{q'}}_\star $

 \item $m=_\Trusted s$ 
 \end{iteMize}
 then 

 \begin{iteMize}{$\bullet$}
 \item $\vec a'_1 = \vec a''_1$
 \item there is $b'_1$ such that 
 \item $d' = c'_1[\vec b'_1]; c_2[\vec b_2]$
 \item $d'' = c''_1[\vec b'_1]; c_2[\vec b_2]$
 \item ${\vec{t''}}_\star   = {\vec{q''}}_\star $
 \item $m' =_\Trusted s'$
 \item $m'' =_\Trusted s''$
 \item $(x,v) = (y,u)$. 

 \end{iteMize}
 \end{lem}

 \begin{proof}
 By induction on the structure of $c_1[\vbullet]$. In the cases of $[\vbullet]$,  $\Skip$, 
and  $x := \Endorse{e}$ 
no declassification events may be produced, so these cases are impossible.

 \begin{iteMize}{$\bullet$}

\item $x := e$. When $D = \emptyset$, no declassification events may be produced. 
When $ D \neq \emptyset$, a   declassification event is produced by
both traces.
Also, $\vec{t''} = \vec{q''} = \epsilon$, and $m' = m$ and $s'= s$. 
Because $m =_\Trusted s$ and
 $ \G (e) = (\_,\Trusted)$ we have that
both traces produces the same declassification event $(x,v)$, and therefore,
 $m'' =_\Trusted s''$.
 \item case $c_\alpha[a_\alpha]; c_\beta[a_\beta]$

 We have two cases depending on whether $c_\alpha[a_\alpha]$ generates low events:

 \begin{enumerate}
 \item

 $\configTwo{c_\alpha[a_\alpha]; (c_\beta [a_\beta]; c_2[a_2])}{m}
 \labarrowstar{}{\ell_1 \cdots \ell_N}
 \configTwo{c_\beta[a_\beta]; c_2[a_2]}{m'}
 $ 
 In this case by Lemma~\ref{lemma:control:seq1} it must be that 
 $
 \configTwo{c_\alpha[b_\alpha]; (c_\beta [b_\beta]; c_2[b_2])}{m} 
 \labarrowstar{}{\ell_1 \cdots \ell_N'}
 \configTwo{c_\beta[b_\beta]; c_2[b_2]}{s'} 
 $
 such that $m' =_\Trusted s'$. Then we can apply the induction hypothesis to $c_\beta[\bullet]$. 
 \item
 In this case $(x,v)$ is produced by $c_\alpha[a_\alpha]$ and we are done by application of induction hypothesis to $c_\alpha[\bullet]$.
 \end{enumerate}

 \item case $\ifThenelse{e}{c_{\cod{true}}}{c_{\cod{false}}}$

 \
 We have two cases:

 \begin{enumerate}
 \item $\G(e) = ( \_, \Trusted)$ 

 In this case both branches take the same command, and we are done by the induction hypothesis.

 \item $\G(e) = (\_, \Untrusted)$. 

 Impossible, because declassification events are not allowed in untrusted integrity contexts.

 \end{enumerate}

 \item case $\while{e}{c_{\cod{loop}}}$

 \

 Similar to sequential composition and conditionals.

 \end{iteMize}

 \end{proof}

 \begin{lem}[Advancement]
 \label{lemma:control:advancement}
 Given 
 \begin{iteMize}{$\bullet$}
 \item program $c_0[\vbullet]$ such that $\G, \pc \vdash c_0[\vbullet]$ 
 \item initial memory $m_0$
 \item two initial attacks $\vec a_0,  \vec b_0$
 \item two intermediate configurations $\configTwo{c[\vec a]}{m}$ and $\configTwo{c[\vec b]}{s}$ such that  
 \item $\configTwo{c_0[\vec a_0]}{m_0} \labarrowstar{}{\vec{t'}} \configTwo{c[\vec a]}{m} \labarrowstar{}{\vec{t''}} 
 \configTwo{c'[\vec a']}{m'} \labarrow{}{(x,v)} \configTwo{c''[\vec a'']}{m''}
  $ 
 \item $\configTwo{c_0[\vec b_0]}{m_0} \labarrowstar{}{\vec{q'}} \configTwo{c[\vec b]}{s} \labarrowstar{}{\vec{q''}}
 \configTwo{d'}{s'} \labarrow{}{(y,u)} \configTwo{d''}{s''}$ 

 \item $\mathrm{PINI}(c_0[\vec a_0], m_0, \vec{t'}, \vec{t'}  \tconcat \vec{t''})$ 

 \item $\mathrm{PINI}(c_0[\vec b_0], m_0, \vec{q'}, \vec{q'} \tconcat \vec{q''})$

 \item $(x,v)$ and $(y,u)$ are declassification events

 \item $\vec b \not\in \Corrupt{c_0[\vbullet], m_0, \vec{t'} \tconcat \vec{t''} \tconcat (x,v) }$

 \item ${\vec{t'}}_\star   = {\vec{q'}}_\star $

 \item $m=_\Trusted s$ 
 \end{iteMize}
 then
 \begin{iteMize}{$\bullet$}
 \item $\vec a' = \vec a''$
 \item there is $b'$ such that 
 \item $d' = c'[\vec b']$
 \item $d'' = c''[\vec b']$
 \item ${\vec{t''}}_\star   = {\vec{q''}}_\star $
 \item $m' =_\Trusted s'$
 \item $m'' =_\Trusted s''$
 \item $(x,v) = (y,u)$. 
 \end{iteMize}
 \end{lem}

 \begin{proof}
By induction on $c[\vbullet]$.   
In the cases of $[\vbullet]$,  $\Skip$,  and  $x := \Endorse{e}$,
no declassification events may be produced, so these cases are impossible.

 \begin{iteMize}{$\bullet$}

\item $x := e$. In case $D = \emptyset$, no declassification events may be produced. 
When $ D \neq \emptyset$, a   declassification event is produced by
both traces.
Also, $\vec{t''} = \vec{q''} = \epsilon$, and $m' = m$ and $s'= s$. 
Because $m =_\Trusted s$ and
 $ \G (e) = (\_,\Trusted)$ we have that
both traces produces the same declassification event $(x,v)$, and therefore,
 $m'' =_\Trusted s''$.

 \item case $c_\alpha; c_\beta$

 By Lemma~\ref{lemma:control:seq2}.

 \item case $\ifThenelse{e}{c_{\cod{true}}}{c_{\cod{false}}}$

 We have two cases:

 \begin{enumerate}

 \item $\G(e) = ( \_, \Trusted)$ 

 In this case both branches take the same command and we are done by induction hypothesis.

 \item $\G(e) = (\_, \Untrusted)$. 

 Impossible, because declassification events are not allowed in untrusted integrity contexts.

 \end{enumerate}

 \item case $\while{e}{c_{\cod{loop}}}$

 Similar to sequential composition and conditionals.\

 \end{iteMize}
 \end{proof}

 \begin{lem}[Control Backbone]
 \label{lemma:control:backbone}
 Given $\G, \pc \vdash c[\bullet]$, memory $m$, an initial attack $\vec a$ and a trace $\vec t $, 
such that
 \begin{multline*}
 \configTwo{c[\vec a]}{m}
 \labarrowstar{}{\vec t_1}
 \configTwo{c_1[\vec a_1]}{m_1}
 \labarrow{}{r_1}
 \configTwo{c'_1[\vec a'_1]}{m'_1}
 \labarrowstar{}{\vec t_{2} \cdots {r_{i-1}}} \\
 \configTwo{c'_{i-1} [\vec{a'}_{i-1}]}{m'_{i-1}}
 \labarrowstar{}{\vec t_i} 
 \framebox{$
 \configTwo{c_i[\vec a_i]}{m_i}
 \labarrow{}{r_i}
 \configTwo{c'_i[\vec{a'}_i]}{m'_i}
 $}
 \labarrowstar{}{}
 \dots
 \end{multline*}
 where $r_i$ are declassification events, then for all 
 $\vec{b}, \vec{q}$ such that 
$(\vec a, \vec t) \sim^{c[\vbullet], m}_\to (\vec b, \vec q) $ and $\vec b \not\in \Corrupt{c[\vbullet], m, \vec t}$, it 
holds that the respective configurations (highlighted in boxes here) match at the 
 declassification events, that is
 \begin{multline*}
 \configTwo{c[\vec b]}{m}
 \labarrowstar{}{\vec q_1}
 \configTwo{c_1[\vec b_1]}{s_1}
 \labarrow{}{r_1}
 \configTwo{c'_1[\vec{b'}_1]}{s'_1}
 \labarrowstar{}{\vec q_2 \cdots r_{i-1}} \\
 \configTwo{c'_{i-1}[\vec{b'}_{i-1}]}{s'_{i-1}} 
 \labarrowstar{}{\vec q_i}
 \framebox{$
 \configTwo{c_i[\vec b_i]}{s_i}
 \labarrow{}{r_i}
 \configTwo{c'_i[\vec{b'}_i]}{s'_i} 
 $}
 \labarrowstar{}{} \dots
 \end{multline*}
where $i$ ranges over the number of declassification events in $\vec t$, and moreover
 \begin{iteMize}{$\bullet$}

 \item
 $m_{i} =_\Trusted s_{i}$  
 and 
 $m'_{i} =_\Trusted s'_{i}$  
 \item
 $\vec {q_i}_\star = \vec{t_i}_\star$

 \end{iteMize}
 \end{lem}

 \begin{proof}
 By induction on the number of declassification events. The base case, where $n = 0$, is immediate. 
For the inductive case,  assume the proposition holds for the first $n$ declassification events in $\vec t$, and apply Lemma~\ref{lemma:control:advancement}.

 \end{proof}
We conclude this section with the proof of Proposition~\ref{prop:control:soundness}.

\

{\bf Proof of Proposition~\ref{prop:control:soundness}}
Consider $\vec b \in R^\triangleright_\to(c[\vbullet], m, \vec a, \vec t) \setminus \Corrupt{c[\vbullet], m, \vec t \tconcat \vec r}$. We want to show 
that $\vec b \in R_\to(c[\vbullet], m, \vec a, \vec t \tconcat \vec r)$. 
Because $\vec b \in R^\triangleright_\to(c[\vbullet], m, \vec a, \vec t)$, we have that 
$\vec b \in 
\{
\vec b\ |\ \exists \vec q~.~(\vec a, \vec t) \sim^{c[\vbullet], m}_\to (\vec b, \vec q) 
~\land 
(
\exists \vec{r'}~.~k_\to(c[\vec b], m_\Public, \vec q_\Public)  \supset
 k(c[\vec b], m_\Public, \vec q_\Public \tconcat \vec{r'}_\Public)
\lor\  \configTwo{c[\vec b]}{m} \Downarrow 
)
\}$. We consider the two cases
\begin{enumerate}[(1)]
\item 
$\vec b \in 
\{
\vec b\ |\ \exists \vec q~.~(\vec a, \vec t) \sim^{c[\vbullet], m}_\to (\vec b, \vec q) 
~\land 
\exists \vec{r'}~.~k_\to(c[\vec b], m_\Public, \vec q_\Public)  \supset
 k(c[\vec b], m_\Public, \vec q_\Public \tconcat \vec{r'}_\Public)
\}
$

\

By definition of $\Corrupt{c[\vbullet], m, \vec t \tconcat \vec r}$, we have that
$\vec b \not\in \Corrupt{c[\vbullet], m, \vec t \tconcat \vec r}
\implies \vec b \not\in \Corrupt{c[\vbullet], m, \vec t}$.
By the Control Backbone Lemma~\ref{lemma:control:backbone}, we have that two traces
agree on the declassification points up to the length of $\vec t$, and
in particular there are $\vec t_0, \vec t_1$, $\vec q_0, \vec q_1$ such that 
$\vec t = \vec t_0 \tconcat \vec t_1 \tconcat \vec r$ and 
$\vec q = \vec q_0 \tconcat \vec q_1$ and that there are no release events
along $\vec t_1$ and $\vec q_1$, for which it holds that 
\begin{multline*}
\configTwo{c [\vec a ]}{m} \labarrowstar{}{\vec t_0} 
\configTwo{c'[\vec a']}{m'} \labarrowstar{}{\vec t_1 \tconcat \vec r}  
\end{multline*}
and 
\begin{multline*}
\configTwo{c [\vec b]}{m} \labarrowstar{}{\vec q_0}
\configTwo{c'[\vec b']}{s'} \labarrowstar{}{\vec q_1 \tconcat \vec{r'}}
\end{multline*}
where $\vec t_\star = \vec q_\star$ and $m' =_\Trusted s'$. By
Advancement Lemma~\ref{lemma:control:advancement}, we obtain that both
traces must agree on $\vec
r$ and $\vec{r'}$. This is sufficient to extend the original
partitioning of $(\vec a, \vec t)$ and $(\vec b, \vec q)$ to 
$(\vec a, \vec t \tconcat \vec r)$ and $ (\vec b, \vec q_0 \tconcat \vec q_1 \tconcat \vec{r'})$
such that 
$(\vec a, \vec t \tconcat \vec r) \sim^{c[\vbullet], m}_\to (\vec b, \vec q_0 \tconcat \vec q_1 \tconcat \vec{r'})$.

\item
 $\vec b \in 
\{
\vec b\ |\ \exists \vec q~.~(\vec a, \vec t) \sim^{c[\vbullet], m}_\to (\vec b, \vec q) 
~\land 
\  \configTwo{c[\vec b]}{m} \Downarrow 
\}$

This case is impossible. By the Control Backbone Lemma~\ref{lemma:control:backbone} there
must be two respective configurations $ \configTwo{c'[\vec a']}{m'}
$ and $\configTwo{c'[\vec b']}{s'}$ where $m'=_\Trusted s'$, such that $\configTwo{c'[\vec
  a']}{m'}$ leads to a release event, but $\configTwo{c'[\vec
  b']}{s'}$ terminates without release events. By analysis of $c'$, 
 similar to the Advancement Lemma, we conclude that none of the
 cases is possible.
\qed
\end{enumerate}

\section{Checked endorsement}
\label{sec:checked}

Realistic applications endorse attacker-provided data based on certain conditions. For
instance, an SQL string that depends on user-provided input is executed if it passes sanitization,
a new password is accepted if the user can provide an old one, and a secret key is accepted
if nonces match.  Because this is a recurring pattern in security-critical
applications, we argue for language support in the form of _checked endorsements_.
 
This section extends the language with checked endorsements and derives both security conditions
and a typing rule for them. Moreover, we show checked
endorsements can be decomposed into a sequence of direct endorsements,
and prove that for well-typed programs, the semantic conditions for
robustness are the same with checked endorsements and with unchecked endorsements.

\paragraph{Syntax and semantics} 
In the scope of this section, we assume checked endorsements are the only endorsement mechanism
in the language. We introduce a syntax for checked endorsements:
\[
c[\vbullet] ::=  \dots\ |\ \EndorseLab{\eta}{x}\ \ifThenelse{e}{c}{c}
\]
The semantics of this command is that a variable $x$ is endorsed if the expression $e$ evaluates to
true. If the check succeeds, the $\cod{then}$ branch is taken, and $x$ is assumed to have high
integrity there.  If the check fails, the $\cod{else}$ branch is taken.
As with direct endorsements, we assume checked endorsements in
program text have unique labels $\eta$. These labels may
be omitted from the examples, but they are explicit in the
semantics.

\draftnote{ Explain that check result is implicitly endorsed.}

\paragraph{Endorsement events}
{Checked endorsement events} 
$\mathit{checked}(\eta, v, b)$ record  the unique label of the endorsement command $\eta$, the value of variable that can potentially be endorsed $v$, and a result of the check $b$, which can be either 
0 or 1.
\begin{mathpar}
\inferrule{m(e) \downarrow v \and v \neq 0}{\configTwo{\EndorseLab{\eta}{x}\
 \ifThenelse{e}{c_1}{c_2}}{m} \labarrow{\mathit{checked}(\eta,m(x),1)}{} \configTwo{c_1}{m}}
\\
\inferrule{m(e) \downarrow v \and v = 0}{\configTwo{\EndorseLab{\eta}{x}\ 
    \ifThenelse{e}{c_1}{c_2}}{m} \labarrow{\mathit{checked}(\eta,m(x),0)}{} \configTwo{c_2}{m}}
\end{mathpar}
\paragraph{Irrelevant attacks} For checked endorsement we define a suitable notion of irrelevant attacks. The reasoning behind this is the following. 
\begin{enumerate}[(1)]
\item Both $\vec t$ and $\vec{t'}$ reach the same endorsement statement: $\eta_i = \eta'_i$.

\item At least one of them results in the positive endorsement: $b_i + b'_i \geq 1$. This ensures
that if both traces do not take the branch then none of the attacks are ignored.

\item The endorsed values are different: $v_i \neq v'_i$. Otherwise, there
should be no further difference in what the attacker can influence along the trace. 

\end{enumerate}
The following definitions formalize the above construction.

\begin{defi}[Irrelevant traces]
\label{def:irrelevant:traces:checked}
Given a trace $\vec t$, where endorsements are labeled as
$\mathit{checked}(\eta_j, v_j,b_j)$, define a set of irrelevant traces based on the number of
checked  endorsements in $\vec t$ as $\altidxcorrupt{i}{\vec t}$.
Then $\altidxcorrupt{0}{\vec t} = \emptyset$, and
\begin{align*}
\begin{split}
 \altidxcorrupt{i}{\vec t} = \{ & \vec{t'}\ |\ \vec{t'} = \vec q\ \tconcat \mathit{checked}(\eta_{i}, v'_i, b'_{i}) \tconcat  \vec{q'}
   \} \mbox {\ such that}\\
   &  \mbox{ $\vec q$ is a prefix of $\vec{t'}$ with $i-1$  $\mathit{checked}$ events, all of which agree with $\mathit{checked}$\ events in $\vec t$}, \\ 
   & \mbox{\ $(b_i + b'_i \geq 1) \land (v_i \neq v'_i)$, and} \\
   & \mbox{ $\vec{q'}$ contains no $\mathit{checked}$ events} \\
\end{split}
\end{align*}
Define \(
\altcorrupt{\vec t} \defn \bigcup_i  \altidxcorrupt{i}{\vec t}  \) as a set of \emph{irrelevant traces} w.r.t. $\vec t$.
\end{defi}

\begin{defi}[Irrelevant attacks ]
\label{def:irrelevant:attacks:checked}
$
\altCorrupt{c[\vbullet], m, \vec t} \defn
\{ \vec{a} |\ 
\configTwo{c[\vec a]}{m} \labarrowstar{}{\vec{t'}} \land
\vec{t'} \in \altcorrupt{\vec t}
\}
$
\end{defi}
Using this definition, we can define security conditions for checked robustness.

\begin{defi}[Progress-sensitive robustness with checked endorsement]
\label{def:robustness:checked:sensitive}
Program  $c[\vbullet]$ satisfies \emph{progress-sensitive robustness
  with checked endorsement} if for all memories $m$ and all attacks $\vec a$, 
such that
$
\configTwo{c[\vec a]}{m} \labarrowstar{}{\vec a} \configTwo{c'}{m'} \labarrowstar{}{\vec r}, 
$ 
and $\vec r$ contains a release event, i.e., 
$k(c[\vec a], m_\Public, \vec t_\Public) \supset k(c[\vec a], m_\Public, \vec t_\Public \tconcat \vec r_\Public)$,
we have \[R^\triangleright (c[\vbullet], m, \vec a, \vec t) \setminus \altCorrupt{c[\vbullet], m, \vec t \tconcat \vec r} \subseteq  R(c[\vbullet], m, \vec a, \vec t \tconcat \vec r)\]

\end{defi}\medskip
\noindent The progress-insensitive version is defined similarly, using progress-insensitive definition for release events and progress-insensitive versions of
 control and release control.
\paragraph{Example} In program
 $
[\bullet]; \EndorseLab{\eta_1}{u}\ \ifThenelse{u = u'}{ \mathit{low} := u  < h}{\Skip}
$,
the attacker can modify $u$ and $u'$. This program is insecure because
the unendorsed, attacker-controlled variable $u'$ influences the decision to declassify.
To see that Definition~\ref{def:robustness:checked:sensitive}
rejects this program, consider running it in
memory with $m(h) = 7$, and two attacks: $a_1$, where attacker sets $u := 5;u' := 0$, and $a_2$,
where attacker sets $u := 5;u' =5$. Denote the corresponding traces up to endorsement by
$\vec t_1$ and $\vec t_2$. We have
$\vec t_1 = [(u,5) \tconcat (u',0)] \tconcat \mathit{checked}(\eta_1, 5, 0) $ and 
$ \vec t_2 = [(u,5) \tconcat (u',5)] \tconcat \mathit{checked}(\eta_1, 5, 1)$.
Because endorsement in the second trace succeeds, this trace also continues with a low event 
$(\mathit{low}, 1)$.
Following Definition~\ref{def:irrelevant:traces:checked} 
we have that 
$t_1 \not\in \altcorrupt{\vec t_2 \tconcat (\mathit{low}, 1)}$, 
implying $a_1 \not \in \altCorrupt{c[\vbullet], m, \vec t_2 \tconcat (\mathit{low}, 1) }$. 
Therefore, $ a_1 \in R^\triangleright (c[\vbullet], m, \vec a_2, \vec t_2 )
\setminus \altCorrupt{c[\vbullet], m, \vec t_2 \tconcat (\mathit{low}, 1) }$. On the other hand, 
$a_1 \not\in R (c[\vbullet], m, \vec a_2, \vec t_2 \tconcat (\mathit{low}, 1))$ because $a_1$ can produce no 
 low events corresponding to $(\mathit{low}, 1)$.

\paragraph{Endorsing multiple variables}
The syntax for checked endorsements can be extended to multiple
variables with the following syntactic sugar,
where  $\eta_i$ is an endorsement label corresponding to variable $x_i$:

\begin{multline*}
\Endorse{x_1,  \dots x_n}\ \ifThenelse{e}{c_1}{c_2} \Longrightarrow 
 \EndorseLab{\eta_1}{x_1}\  \ifThenelse{e} {  \\ \EndorseLab{\eta_2}{x_2}\ \ifThenelse{\cod{true} }{ \dots  c_1 }{\Skip \dots } }{c_2}
\end{multline*}
Note that in this encoding the condition is checked as early as possible; an
alternative encoding here would check the condition in the end. While such encoding would have an
advantage of type checking immediately, we believe that checking the condition
as early as possible avoids spurious (albeit harmless in this simple context) endorsements of all
but the last variable, and is therefore more faithful semantically.

\paragraph{Typing checked endorsements} 
To enforce programs with checked endorsements, we extend the type
system with the following general rule:
\begin{mathpar}
\small
\inferrule[(T-CHECKED)]{
\G' \defn \G [x_i \mapsto \G( x_i ) \sqcap (\High, \Trusted ) ]  \and
\G' \vdash e: \level', D'
\and
\pc' \defn \pc \join \level'
\\
 \pc' \sqsubseteq (\High, \Trusted) \and
\and
\G' , \pc' \vdash c_1 
\and
\G, \pc' \vdash c_2 
}{
\G, \pc \vdash \Endorse{x_1, \dots, x_n}\  \ifThenelse{e}{c_1}{c_2}
}
\end{mathpar}
The expression $e$ is 
type-checked in an environment $\G'$ in which endorsed variables $x_1,
\dots x_n$ have trusted integrity;
its label $\level'$  is joined to form auxiliary $\pc$-label $\pc'$. The level of 
$\pc'$ must be trusted, ensuring that endorsements happen in a trusted context, and
that no declassification in $e$ depends on untrusted variables other
than the $x_i$ (this effectively subsumes the need to check individual variables in $D'$).
Each of the branches is type-checked
with the program label set to $\pc'$;
however, for $c_1$ we use the auxiliary typing environment $\G'$,
since the $x_i$ are trusted there.

Program
 $
[\bullet]; \Endorse{u}\ \ifThenelse{u = u'}{ \mathit{low} := \decl{u  < h}}{\Skip}
$
is rejected by this type system. Because variable $u'$ is not endorsed, the auxiliary $\pc$-label has untrusted integrity.

\draftnote{Do we need more examples here?}

\subsection{Relation to direct endorsements}
Finally, for well-typed programs we can safely translate checked endorsements to direct
endorsements using a translation in which a checked endorsement of $n$ variables is translated to 
$n+1$ direct endorsements. First, we unconditionally endorse the
result of the check. The rest of the endorsements happen in the
$\cod{then}$ branch, before translation of $c_1$.
We save the results of the endorsements in temporary variables $t_1
\dots t_n$
and replace all occurrences of $x_1 \dots x_n$ within $c_1$ with the
temporary ones 
(we assume that each $t_i$ has the same
confidentiality level as the corresponding original $x_i$, and $t_0$ has the
confidentiality level of the expression $e$). All other commands are translated to themselves.

\begin{defi}[Labeled translation from checked endorsements to direct endorsements]
\label{def:endorse:translation}
Given a program $c[\vbullet]$ that only uses checked endorsements, we define its labeled translation to direct endorsements $\trans{c[\vbullet]}$  
 inductively:

\begin{iteMize}{$\bullet$}

\item

$
 \trans{\EndorseLab{\eta}{x_1, \dots x_n}\ \ifThenelse{e}{c_1}{c_2} } \Longrightarrow t_0 := \EndorseLab{\eta_0}{e};
 \ifThenelse{t_0\\ ~~~~~~~~ }{ t_1:= \EndorseLab{\eta_1}{x_1}; \dots t_n := \EndorseLab{\eta_n}{x_n};  \trans{c_1 [t_i/x_i] }        
}{  \trans {c_2}}
$

\item
$\trans{c_1; c_2} \Longrightarrow \trans{c_1}; \trans{c_2} $

\item
$\trans{\ifThenelse{e}{c_1}{c_2}}  \Longrightarrow  \ifThenelse{e}{ \trans{c_1}}{\trans{c_2}}$ 

\item
$\trans{\while{e}{c}} \Longrightarrow \while{e}{\trans c} $ 

\item
$\trans c \Longrightarrow c $, for other commands $c$. 

\end{iteMize}

\end{defi}

\paragraph{Adequacy of translation for checked endorsements for well-typed programs}
Next we show adequacy of the labeled translation of
Definition~\ref{def:endorse:translation} for well-typed programs. Note that for non-typed programs this
adequacy does not hold, as shown by an example in the end of the section. 

Without loss of generality, we assume checked endorsements
have only one variable ($n=1$ in the translation of checked endorsement in
Definition~\ref{def:endorse:translation}). We adopt an indexing convention where
 checked endorsement with the label $\eta_i$, corresponds to two direct endorsements 
with the labels $\eta_{2i - 1}$ and $ \eta_{2i}$. 
The following lemma establishes a connection between irrelevant attacks of the source and translated runs.

\begin{lem}[Synchronized endorsements]
\label{lemma:translation:sync}
Given a program $c[\vbullet]$ that only uses checked endorsements, such that $\G, \pc \vdash c[\vbullet]$,  memory $m$, and attack $\vec a$,
such that
\[
 \configTwo{c[\vec a]}{m} \labarrowstar{}{\vec t} 
\mbox{\ and\ } 
\configTwo{\trans{c[\vec a]}}{m} \labarrowstar{}{\hat{\vec t}} 
\]
where 
\begin{iteMize}{$\bullet$}
\item
$\vec t = \vec{t'} \tconcat \mathit{checked} (\eta_i, v_i, b_i) $ and 
\item
  $\hat{\vec t} = \hat{\vec{t'}} \tconcat \mathit{endorse} (\eta_{2i-1}, 0) $ or $\hat{\vec t} = \hat{\vec{t'}} \tconcat \mathit{endorse} (\eta_{2i-1}, 1) \tconcat \mathit{endorse} (\eta_{2i}, v) $ 
\item $k$ is a number of checked endorse events in $\vec t$ and
\end{iteMize}
we have that 
\begin{iteMize}{$\bullet$}

\item
$R(c[\vbullet], m, \vec a, \vec t) = R(\trans{c[\vbullet]}, m, \vec a,
\hat{\vec t})$.  

\item
$R_\to(c[\vbullet], m, \vec a, \vec t) = R_\to(\trans{c[\vbullet]}, m, \vec a,
\hat{\vec t})$.

\item
$\Corrupt{\trans{c[\vbullet]}, m, \hat{\vec t}} = \altCorrupt{c[\vbullet], m, \vec t}$

\end{iteMize}
\end{lem}
\begin{proof}

The first two items follow from the definition of the translation, because the translation does not generate new release events.

To prove the second item, we consider partitions of irrelevant traces generated by every $k$-th checked endorsement
 and the direct endorsement(s) that correspond to it. We proceed by induction on $k$. 
For the base case, $k= 0$, i.e.,  neither $\vec t$ nor $\hat{\vec t}$ 
contain endorsements, it holds that
$\Corrupt{\trans{c[\vbullet]}, m, \hat{\vec t}} = \altCorrupt{c[\vbullet], m, \vec t} = \emptyset$. 
For the inductive case, define a pair of auxiliary sets 
\begin{align*}
 &F_k \defn \Corrupt{\trans{c[\vbullet]}, m, \hat{\vec t}} \setminus 
\Corrupt{\trans{c[\vbullet]}, m, \hat{\vec{t'}}}\\
&  P_k \defn \altCorrupt{c[\vbullet], m, {\vec t}} \setminus 
\altCorrupt{c[\vbullet], m, {\vec{t'}}}
\end{align*}
By the induction hypothesis, 
$\Corrupt{c[\vbullet], m, \hat{\vec{t'}}} = \altCorrupt{c[\vbullet], m, \vec{t'}}$. 
By Definitions~\ref{def:irrelevant:attacks}
and~\ref{def:irrelevant:attacks:checked}, we know that
$\Corrupt{\trans{c[\vbullet]}, m, \hat{\vec t}} \supseteq 
\Corrupt{\trans{c[\vbullet]}, m, \hat{\vec{t'}}}$ and
$\altCorrupt{c[\vbullet], m, {\vec t}} \supseteq
\altCorrupt{c[\vbullet], m,{\vec{t'}}}$. Therefore, in order to prove that
$\Corrupt{\trans{c[\vbullet]}, m, \hat{\vec t}} = \altCorrupt{c[\vbullet], m,{\vec t}}$ it is sufficient to show 
that $F_k = P_k$. We  consider 
each direction of equivalence separately.
\begin{iteMize}{$\bullet$}
\item
$F_k \supseteq P_k$. 
Take  an attack $\vec b \in P_k$. That is
$\configTwo{c[\vec b]}{m}$ produces a trace $\vec q$ such that $\vec
q$ agrees on all checked endorsements with $\vec t$ except the last
one. There are three possible ways in which these endorsements may disagree:

\begin{enumerate}

\item
Trace $\vec t$ contains $\mathit{checked} (\eta_k, v_k, 1)$ and 
$\vec q$ contains $\mathit{checked} (\eta_k, v'_k, 1)$ such that $v_k \neq v'_k$. 
By the rules for the translation, it must be that the trace $\hat{\vec t}$, which is produced
by configuration $\configTwo{\trans{c[\vec a]}}{m}$, has two corresponding
endorsement events
$\mathit{endorse}(\eta_{2k - 1}, 1)$ and 
$\mathit{endorse}(\eta_{2k}, v_k)$. Similarly, the trace $\hat{\vec
  q}$,  produced by $\configTwo{\trans{c[\vec b]}}{m}$, has two 
corresponding endorsement events
$\mathit{endorse}(\eta_{2k - 1}, 1)$ and
$\mathit{endorse}(\eta_{2k}, v'_k)$. Because $v'_k \neq v_k$ we have
that $\hat{\vec q} \in \corrupt{\hat{\vec t}}$.

\item

Trace $\vec t$ contains checked endorsement event 
$\mathit{checked}(\eta_k, v_k, 1)$, 
while trace $\vec q$ contains event $\mathit{checked}(\eta_k, v'_k,
0)$. 
In this case, the trace $\hat{\vec t}$ obtained from running
$\configTwo{\trans{c[\vec a]}}{m}$ must contain two endorsement events
$\mathit{endorse}(\eta_{2k-1}, 1)$ 
and
$\mathit{endorse}(\eta_{2k}, v_k)$, while 
the trace $\hat{\vec q}$ corresponding to 
$\configTwo{\trans{c[\vec b]}}{m}$ contains one event
$\mathit{endorse}(\eta_{2k-1}, 0)$. Therefore, 
$\hat{\vec q} \in \corrupt{\hat{\vec t}}$. 

\item 
Trace $\vec t$ contains checked endorsement event 
$\mathit{checked}(\eta_k, v_k, 0)$, 
while trace $\vec q$ contains event $\mathit{checked}(\eta_k, v'_k,
1)$. This is similar to the previous case. 
\end{enumerate}

From $\hat{\vec q} \in \corrupt{\hat{\vec t}}$ it follows that 
$\vec b \in F_k$.

\item 
$F_k \subseteq P_k$. Take an attack $\vec b \in F_k$. There must be
a trace $\hat{\vec q}$,  produced by $\configTwo{\trans{c[\vec
    b]}}{m}$, 
such that
$\hat{\vec q} \in \corrupt{\hat{\vec t}}$.
There are two ways this can happen:
\begin{enumerate}
\item
$\hat{\vec q}$ and $\hat{\vec t}$ disagree at the translated
endorsement event that has label $\eta_{2k-1}$. More precisely, one
must have form $\mathit{endorse}(\eta_{2k-1}, 1)$ and
the other, $\mathit{endorse}(\eta_{2k-1}, 0)$.
In the original run, this corresponds to two traces $\vec t$ and $\vec
q$ such that $\vec t$ contains the event $\mathit{checked}(\eta_k, b_k,
v_k)$ and $\vec q$ contains the event $\mathit{checked}(\eta_k, b'_k,
v'_k)$. 
We know that $b_k = 1$ and $b'_k = 0$, and hence $ b_k + b'_k \leq 1$.
According to
Definition~\ref{def:irrelevant:traces:checked}, we need to show
that $v'_k \neq v_k$. Assume this is not the case, and 
that $v'_k = v_k$.  Then by rule (T-CHECKED), we have $b_k = b'_k$,
which contradicts the earlier conclusion. Hence $\vec q \in \altcorrupt{\vec t}$. 

\item Alternatively, $\hat{\vec q}$ and $\hat{\vec t}$ disagree at
  the endorsement event that has label $\eta_{2k}$. This also means that
  they agree on the earlier endorsement, i.e., for the corresponding
  trace with checked endorsement, we can show that $b_k = b'_k =
  1$, and  $v_k \neq v'_k$. Therefore, $\vec q \in \altcorrupt{\vec
    t}$. 
\end{enumerate}

From $\vec q \in \altcorrupt{\vec t}$ it follows that $\vec b \in
P_k$. 

\end{iteMize}

\end{proof}\smallskip

\noindent Using Lemma~\ref{lemma:translation:sync} we can show the following
Proposition, which
relates the security of the source and translated programs.
\begin{prop}[Relation of checked and direct endorsements]
\label{prop:checkedtodirect}
Given a program $c[\vbullet]$ 
that only uses checked endorsements such that $\G, \pc \vdash c[\vbullet] $, then 
 $c[\vbullet]$ satisfies progress-insensitive robustness for checked endorsements
if and only $\trans{c[\vbullet]}$ satisfies progress-insensitive robustness for direct endorsements. 
\end{prop}

\begin{proof}
Note that our translation preserves typing:
when $\G, \pc \vdash c[\vbullet]$, then
 $\G, \pc
\vdash \trans{c[\vbullet]}$. Therefore, by
Proposition~\ref{prop:control:soundness} 
the translated program satisfies progress-insensitive robustness with endorsements. 
To show that the source program satisfies the progress-insensitive 
robustness with checked endorsements, we use
Lemma~\ref{lemma:translation:sync} and note that the corresponding sets
of irrelevant attacks and control between any two runs of the programs
must be in sync.
\end{proof}

\paragraph{Notes on the adequacy of the translation} We observe two facts about the adequacy of this
translation. First,
for non-typed programs, the relation does not hold. For instance,
a program like \[ 
[\bullet]; \Endorse{u}\ \ifThenelse{u = u'}{
  \mathit{low} := \decl{u < h}}{\Skip} \]
does not satisfy Definition~\ref{def:robustness:checked:sensitive}.
However, translation of this program satisfies Definition~\ref{def:robustness:qualified:sensitive}. 

Second, observe that omitting endorsement of the expression would lead to occlusion. Consider an
alternative translation that endorses only the variables $x_1, \dots
x_n$ but not the result of the whole expression. Using such a translation, a program
\[\ifThenelse{u \cdot 0  > 0}{\Skip}{\Skip}; \mathit{trusted}:=x\]
is translated to 
\[\mathit{temp} := x; \ifThenelse{ t \cdot 0 > 0}{\Skip}{\Skip}; \mathit{trusted} := x\] However,
while the first program does not satisfy Definition~\ref{def:robustness:checked:sensitive}, 
the second program is accepted by Definition~\ref{def:robustness:qualified:sensitive}.

\section{Attacker impact}
\label{sec:attacker-impact}

In prior work, robustness controls the attacker's ability to
cause information release. In the presence of endorsement, the
attacker's ability to influence trusted locations also becomes
an important security issue.
To capture this influence,  we introduce an integrity dual
to attacker knowledge, called _attacker impact_.
Similarly to low events, we define \emph{trusted events}
as assignments to trusted variables and termination.

\begin{defi}[Attacker impact ]
  Given a program $c[\vbullet]$, memory $m$, and trusted events  $\vec t_\star$, define $p
  (c[\vbullet], m, \vec t_\star)$ to be a set of attacks $\vec a$ that match
  trusted events  $\vec t_\star$:
\[
p(c[\vbullet], m, \vec t_\star ) \defn \{\vec a \ |\ \configTwo{c[\vec a]}{m} 
\labarrowstar{}{\vec{t'}} \land \vec t_\star = \vec{t'}_\Trusted \}
 \]\vspace{-3 pt}
\end{defi}

\noindent Attacker impact is defined with respect to a given sequence of trusted events $\vec t_\star$, starting
 in memory
$m$, and program $c[\vbullet]$. The impact is the set of all attacks that agree with $\vec t_\star$ in
their footprint on trusted variables.  

Intuitively, a smaller set for attacker impact means that the attacker has greater power to influence trusted events. 
\draftnote{ Draw further parallels between knowledge and impact. }
Similarly to progress knowledge, we define _progress impact_,
characterizing which attacks lead to one more trusted event.
This then allows us to define robustness conditions for _integrity_,
which have not previously been identified.
\begin{defi}[Progress impact]
Given a program $c[\vbullet]$, memory $m$, and  sequence of trusted
events  $\vec t_\star$, define progress impact $p_\to(c[\vbullet], m, \vec t_\star)$ as
\[
p_\to(c[\vbullet], m, \vec t_\star) 
\defn 
\{
\vec a\ |\ 
\configTwo{c[\vec a]}{m}
\labarrowstar{}{\vec{t'}}
\configTwo{c'}{m'}
\land \vec t_\star = \vec{t'}_\Trusted
\land
\configTwo{c'}{m'}
\labarrowstar{}{t''_\star}
\}
\]\vspace{-5 pt}
\end{defi}
\noindent The intuition for the baseline robustness for integrity  is that attacker should
not influence trusted data. This is similar to noninterference for
integrity (modulo availability attacks, which have not been explored in
this context before). However unlike earlier work, we can easily
extend the notion of integrity robustness to
endorsements and checked endorsements.

\begin{defi}[Progress-insensitive integrity robustness with endorsements]
A program $c[\vbullet]$ satisfies progress-insensitive robustness for integrity if 
for all memories $m$, and for all traces $\vec t \tconcat t_\star$ where $t_\star$ is a trusted event,
we have 
\[
p_\to(c[\vbullet], m, \vec t_{\Trusted}) \setminus  \Corrupt{c[\vbullet], m, \vec t \tconcat t_\star} \subseteq  p(c[\vbullet], m, \vec t_{\Trusted} \tconcat t_\star) 
\]
\end{defi}\vspace{7 pt}

\noindent Irrelevant attacks are defined precisely as in Section~\ref{sec:endorsement}.
We omit the corresponding definitions for programs without endorsements and with checked endorsements.\newpage

The type system of Section~\ref{sec:enforcement} also enforces 
integrity robustness with endorsements, rejecting insecure
programs such as $t := u$ and $"if"~(u_1)~"then"~t := "endorse"(u_2)$,
but accepting $t := "endorse"(u)$.
Moreover, a connection between
checked and direct endorsements, analogous to
Proposition~\ref{prop:checkedtodirect}, holds for integrity robustness too.

\section{Examples}
\label{sec:examples}

\paragraph{Password update}
Figure~\ref{fig:example:passwordupdate} shows code for 
updating a password. The attacker controls variables
$\cod{guess}$ of level $(\Public, \Untrusted)$  and
$\cod{new\_password}$ of level $(\Secret, \Untrusted)$. 
The variable $\cod{password}$ has level $(\Secret, \Trusted)$ and
variables $\cod{nfailed}$ and $\cod{ok}$ have level
$(\Public,\Trusted)$.
The declassification on line~\ref{example:password:declassify} uses
the untrusted variable $\cod{guess}$. This variable, however, is
listed in the $\cod{endorse}$ clause on
line~\ref{example:password:endorse}; therefore, the declassification
is accepted.  The initially untrusted variable
$\cod{new\_password}$ has to be endorsed to update the password on
line~\ref{example:password:update}. The example also shows how other
trusted variables---$\cod{nfailed}$ and $\cod{ok}$---can be updated in
the $\cod{then}$ and $\cod{else}$ branches. 

\vspace{-0.5ex}
\paragraph{Data sanitization}
Figure~\ref{fig:embargoed} shows an annotated version of the code from
the introduction, in which some information ("new_data") is not
allowed to be released until time "embargo_time".
The attacker-controlled variable is "req_time" of 
level $(\Public, \Untrusted)$, and "new_data" has level $(\Secret,
\Trusted)$. The checked
endorse ensures that the attacker cannot violate the integrity
of the test "req_time >= embargo_time". (Variable
"now" is high-integrity and contains the current
time). Without the checked endorse, the release of "new_data" would
not be permitted either semantically or by the type system.

\begin{figure}[t]%
\lstset{numbers=left, xleftmargin=20pt, framexleftmargin=20pt}
\begin{minipage}{0.48\textwidth}
\begin{lstlisting}
[$\bullet$]
endorse(guess, new_password) /*@ \label{example:password:endorse} @*/
 if (declassify(guess==password)) /*@ \label{example:password:declassify} @*/
  then 
    password = new_password;  /*@ \label{example:password:update} @*/
    nfailed  = 0;
    ok = true;
  else
    nfailed  = nfailed + 1; 
    ok = false;
\end{lstlisting}
\caption{Password update}
\label{fig:example:passwordupdate}
\end{minipage}
~
\begin{minipage}{0.50\textwidth}
\vspace{1.2em}
\begin{lstlisting}
[$\bullet$]
endorse(req_time)
if (req_time <= now)
 then 
  if (req_time >= embargo_time)
   then return declassify(new_data)
   else return old_data
 else
   return old_data
\end{lstlisting}
\caption{Accessing embargoed information}
\label{fig:embargoed}
\end{minipage}
\end{figure}

\section{Related work} 
\label{sec:related}

Prior robustness definitions~\cite{Myers:Sabelfeld:Zdancewic:JCS06,Chong:Myers:CSFW06},
based on equivalence of low traces, do
not differentiate programs such as $[\bullet ]; \mathit{low} := u <
h; \mathit{low'} := h$ and $[\bullet]; \mathit{low'} := h;
\mathit{low} := u <h$;  
Per dimensions of information release~\cite{Sabelfeld:Sands:JCS}, the
new security conditions cover not only the ``who'' dimension, but are also
  sensitive to ``where'' information release happens. Also, the security condition
  of robustness with endorsement does not suffer from the occlusion problems of
  qualified robustness.
 Balliu and Mastroeni~\cite{Balliu:Mastroeni:PLAS09}
 derive sufficient conditions for robustness using weakest
 precondition semantics.  These conditions are not precise
 enough to distinguish the examples above, and, moreover, do not
 support endorsement.

Prior work on robustness semantics defines
termination-insensitive security
conditions \cite{Myers:Sabelfeld:Zdancewic:JCS06,Balliu:Mastroeni:PLAS09}.
Because the new framework is powerful enough to capture
the security of programs with intermediate observable events, it
can describe the robustness of nonterminating programs. Prior work
on qualified robustness~\cite{Myers:Sabelfeld:Zdancewic:JCS06}
uses a non-standard _scrambling_ semantics in which
qualified robustness unfortunately becomes a _possibilistic_
condition, leading to 
anomalies such as reachability of dead code.
The new framework avoids such artifacts because it
uses a standard, deterministic semantics.

Checked endorsement was introduced informally in the Swift web application
framework~\cite{Chong+:SOSP07} as a convenient way
to implement complex security policies. The current paper is the first
to formalize and to study the properties of checked endorsement.

Our semantic framework is based on the definition of attacker
knowledge, developed in prior work introducing _gradual release_~\cite{Askarov:Sabelfeld:SP07}.
Attacker knowledge is used for expressing confidentiality
policies in recent work~\cite{Banerjee+:SP08,Askarov+:Termination,Askarov:Sabelfeld:CSF09,Broberg:Sands:PLAS09}. 
However, none of this work considers integrity;
applying attacker-centric reasoning to
integrity policies is novel.

\section{Conclusion}
\label{sec:conclusion}

We have introduced a new knowledge-based framework for semantic
security conditions for information security with declassification and
endorsement.  A key technical innovation is characterizing the impact and
control of the attacker over information in terms of sets of similar
attacks. Using this framework, we can express semantic conditions that
more precisely characterize the security offered by a security type
system, and derive a satisfactory account of new language features
such as checked endorsement.
 
\bibliographystyle{custom}
\bibliography{literature,pm-master}

\section*{Acknowledgments}

The authors would like to thank the anonymous reviewers
for comments on a draft of this paper. We also thank Owen Arden,
Stephen Chong, Michael Clarkson, Daniel Hedin, Andrei Sabelfeld, and
Danfeng Zhang for useful discussions.

This work was supported by a grant from the Office of Naval Research
(N000140910652) and by two NSF grants (the TRUST center, 0424422;
and 0964409).  The U.S. Government is authorized to reproduce and
distribute reprints for Government purposes, notwithstanding any
copyright annotation thereon.  The views and conclusions contained
herein are those of the authors and should not be interpreted as
necessarily representing the official policies or endorsement, either
expressed or implied, of any of the funding agencies or of the U.S.\@
Government.

\end{document}

